\documentclass[11pt,A4]{article}
\usepackage{amsmath,amsfonts,amssymb,amsthm}
\usepackage{mathrsfs}
\usepackage{latexsym}
\usepackage[english]{babel}

\usepackage{graphicx}
\usepackage[usenames,dvipsnames]{xcolor}

\renewenvironment{proof}[1][\proofname]{{\bfseries #1.} }{\qed}

\newcommand{\field}[1]{\mathbb{#1}}

\newcommand{\R}{\field{R}}

\newcommand{\N}{\field{N}}

\newcommand{\Z}{\field{Z}}

\newcommand{\Var}{{\rm Var}}

\newcommand{\e}{{\rm e}}

\newcommand{\Mc}{{\mathcal M}}

\newcommand{\for}{\; \mbox{for}\;}

\newcommand{\Tb}{{\mathbb{T}}}
\newcommand{\Ec}{{\mathcal{E}}}
\newcommand{\Lc}{{\mathcal{L}}}
\newcommand{\var}{\operatorname{Var}}
\newcommand{\Sc}{\mathcal{S}}
\newcommand{\Nc}{\mathcal{N}}

\def\authors#1{{ \begin{center} #1 \vspace{0pt} \end{center} } \smallskip}
\def\institution#1{{\sl \begin{center} #1 \vspace{0pt} \end{center} } }
\def\inst#1{\unskip $^{#1}$}
\def\title#1{{\huge\bf  \begin{center} #1 \vspace{0pt} \end{center}  } \smallskip}

\def\E{{\mathbb{ E}}}

\def\P{{\mathbb{P}}}

\def\paref#1{(\ref{#1})}

\newtheorem{theorem}{Theorem}[section]
\newtheorem{proposition}[theorem]{Proposition}
\newtheorem{lemma}[theorem]{Lemma}

\newtheorem{defn}[theorem]{Definition}

\newtheorem{remark}[theorem]{Remark}
\begin{document}

\title{{\textsc{Non-Universality \\ of Nodal Length Distribution \\for Arithmetic
Random Waves}}}
\date{July 10, 2015}
\authors{Domenico Marinucci%\symbolfootnote[1]{Corresponding author: e-mail address of the corresponding author}
\inst{1}, Giovanni Peccati\inst{2}, Maurizia Rossi\inst{2} and Igor Wigman\inst{3}}

\institution{\inst{1}Dipartimento di Matematica, Universit\`a di Roma Tor Vergata\\
\inst{2}Unit\'e de Recherche en Math\'ematiques, Universit\'e du Luxembourg\\
\inst{3}Department of Mathematics, King's College London}

%\maketitle
\begin{abstract} ``Arithmetic random waves" are the Gaussian Laplace eigenfunctions on
the two-dimensional torus \cite{RudWig,KKW}. In this paper we find that their nodal length converges to a non-universal (non-Gaussian) limiting distribution,
depending on the angular distribution of lattice points lying on circles.

Our argument has two main ingredients. An explicit derivation of the Wiener-It\^o chaos expansion for the nodal length shows that
it is dominated by its $4$th order chaos component (in particular, somewhat surprisingly, the second order chaos component vanishes).
The rest of the argument relies on the precise analysis of the fourth order chaotic component.

\smallskip

\noindent\textbf{Keywords and Phrases: }Arithmetic Random Waves, Nodal Lines, Non-Central Limit Theorem, Berry's Cancellation.

\smallskip

\noindent \textbf{AMS Classification:} 60G60, 60D05, 60B10, 58J50, 35P20

\end{abstract}

\section{Introduction and main results}

\subsection{Arithmetic random waves}

Let $\Tb:=\R^2/\Z^2$ be the standard
$2$-torus and $\Delta$ the Laplacian on $\Tb$. We are interested
in the (totally discrete) spectrum of $\Delta$ i.e., eigenvalues $E>0$
of the Schr\"{o}dinger equation
\begin{equation}
\label{eq:Schrodinger}
\Delta f + Ef=0.
\end{equation}
Let $$S=\{{{} n \in \Z : n} =  a^2+b^2 \,\,  \mbox{{} for some} \:a,\, b\in\Z\}$$ be the collection of all numbers
expressible as a sum of two squares. Then, the eigenvalues of \eqref{eq:Schrodinger}
(also called {\it energy levels} of the torus) are all numbers of the form $E_{n}=4\pi^{2}n$ with $n\in S$.

In order to describe the Laplace eigenspace corresponding to $E_{n}$, denote by
$\Lambda_n$ the set of {\it frequencies}:
\begin{equation*}
\Lambda_n := \lbrace \lambda =(\lambda_1,\lambda_2)\in \Z^2 : \lambda_1^2 + \lambda_2^2 = n\rbrace\
\end{equation*}
whose cardinality
\begin{equation}
\label{eq:Nn=|Lambda|}
\mathcal N_n := |\Lambda_n|=r_{2}(n)
\end{equation}
equals the number of ways to express $n$ as a sum of two squares.
(Geometrically, $\Lambda_{n}$ is the collection of all standard lattice points
lying on the centred circle with radius $\sqrt{n}$.)
For $\lambda\in \Lambda_{n}$ denote the complex exponential associated to the frequency $\lambda$
\begin{equation*}
e_{\lambda}(x) = \exp(2\pi i \langle \lambda, x \rangle)
\end{equation*}
with $x=(x_{1},x_{2})\in\Tb$.
The collection
\begin{equation*}
\{e_{\lambda}(x)\}_{\lambda\in \Lambda_n}
\end{equation*}
of the complex exponentials corresponding to the frequencies $\lambda\in\Lambda_{n}$,
is an $L^{2}$-orthonormal basis of the eigenspace $\Ec_n$ of $\Delta$ corresponding to the
eigenvalue $E_{n}$. In particular, the dimension of $\Ec_{n}$
is
\begin{equation*}
\dim \Ec_{n}=\mathcal N_n = |\Lambda_n|
\end{equation*}
(cf. \eqref{eq:Nn=|Lambda|}).
The number
$\Nc_{n}$ is subject to large and erratic fluctuations; it grows \cite{La} {\em on average}
as $\sqrt{\log{n}}$, but could be as small as $8$ for (an infinite sequence of) prime numbers
$p\equiv 1\mod{4}$, or as large as a power of $\log{n}$.

Following ~\cite{RudWig} and ~\cite{KKW}, we define
the {\it arithmetic random waves} (also called {\it random Gaussian toral Laplace eigenfunctions})
to be the random fields
\begin{equation}\label{defrf}
T_n(x)=\frac{1}{\sqrt{\mathcal N_n}}\sum_{ \lambda\in \Lambda_n}a_{\lambda}e_\lambda(x), \quad x\in \Tb,
\end{equation}
where the coefficients $a_{\lambda}$ are standard complex-Gaussian random variables verifying the following properties: $a_\lambda$ is stochastically independent of $a_\gamma$ whenever $\gamma \notin \{\lambda, -\lambda\}$, and
$$a_{-\lambda}= \overline{a_{\lambda}}$$ (ensuring that the $T_{n}$ are real-valued).\footnote{From now on, we assume that every random object considered in this paper is defined on a common probability space $(\Omega, \mathcal{F}, \P)$, with $\E$ denoting mathematical expectation with respect to $\P$.}
By the definition \eqref{defrf}, $T_n$ is a stationary
(i.e. the law of $T_{n}$ is invariant under all the translations $$f(\cdot)\mapsto f(x'+\cdot),$$
$x'\in \Tb$), centered
Gaussian random field with covariance function
\begin{equation*}
r_n(x,x') = r_{n}(x-x') := \E[T_n(x) \overline{T_n(x')}] = \frac{1}{\mathcal N_n}
\sum_{\lambda\in \Lambda_n}e_{\lambda}(x-x')=\frac{1}{\mathcal N_n}\sum_{\lambda\in \Lambda_n}\cos\left(2\pi\langle x-x',\lambda \rangle\right),
\end{equation*}
$x,x'\in\Tb$ (by the standard abuse of notation for stationary fields). Note that $r_{n}(0)=1$, i.e. $T_{n}$ has unit variance.

\subsection{Nodal length: mean and variance}

Consider the total {\it nodal length} of the random
eigenfunctions, i.e. the collection $\{\Lc_{n}\}_{n\in S}$ of all random variables with the form
\begin{equation}\label{e:length}
\Lc_n := \text{length}(T_n^{-1}\lbrace 0 \rbrace).
\end{equation}
The expected value of $\Lc_{n}$ was computed in \cite{RudWig} to be
\begin{equation}
\E[\mathcal L_n]= \frac{1}{2\sqrt{2}}\sqrt{E_n},
\end{equation}
consistent with Yau's conjecture ~\cite{Yau,DF}.
The more challenging question of the asymptotic behaviour of the variance $\var(\Lc_{n})$ of $\Lc_{n}$ was addressed in \cite{RudWig}, and fully resolved in \cite{KKW} as follows.

Given $n\in S$, define a probability measure $\mu_{n}$ on the unit circle $\Sc^{1}\subseteq\R^{2}$
supported on angles corresponding to lattice points in $\Lambda_{n}$:
\begin{equation*}
\mu_{n} := \frac{1}{\mathcal N_n} \sum_{\lambda\in \Lambda_n} \delta_{\frac{\lambda}{\sqrt{n}}}.
\end{equation*}
It is known ~\cite{EH} that for a density $1$ sequence of numbers $\{n_{j}\}\subseteq S$ the angles
of lattice points in $\Lambda_{n}$ tend to be equidistributed, in the sense that
\begin{equation}
\label{eq:mun equidist}
\mu_{n_{j}}\Rightarrow \frac{d\phi}{2\pi}
\end{equation}
(where $\Rightarrow$ indicates weak-$*$ convergence of probability measures, and $d\phi$ stands for the Lebesgue measure on
$\Sc^{1}$). However the sequence $\{\mu_{n}\}_{n\in S}$
has other weak-$*$ adherent points ~\cite{Ci,KKW} (called {\it attainable measures}), partially classified in \cite{KW}.

It was proved in ~\cite{KKW} that one has
\begin{equation}
\label{eq:var leading KKW}
\var(\Lc_{n}) =c_n \frac{E_n}{\Nc_{n}^2}(1 + o_{\Nc_{n}\rightarrow\infty}(1)),
\end{equation}
where
\begin{equation}\label{cn}
c_n = \frac{1+\widehat{\mu_n}(4)^2}{512},
\end{equation}
and, for a measure $\mu$ on $\Sc^{1}$,
\begin{equation*}\label{e:smet}
\widehat \mu(k) = \int_{\Sc^{1}} z^{-k}\,d\mu(z), \quad k\in \mathbb{Z},
\end{equation*}
are the Fourier coefficients of $\mu$ on the unit circle.
As $$|\widehat{\mu_{n}}(4)|\le 1$$ by the triangle inequality, the result \paref{eq:var leading KKW}
shows that the true order of magnitude of $\var(\Lc_{n})$ is
$ \frac{E_n}{\Nc_n^2} $: this is of smaller order than what would be a natural guess, namely
$ \frac{E_n}{\Nc_n} $; this situation (customarily called {\it arithmetic Berry's cancellation}, see \cite{KKW}) is similar to the {\it cancellation phenomenon} observed by Berry in a different setting, see \cite{Berry 2002,wig}.

In addition, \eqref{eq:var leading KKW} shows that, in order for $\var(\Lc_{n})$ to exhibit an
asymptotic law (equivalent to $\{c_{n}\}$ in \eqref{cn} being convergent along a subsequence)
we need to pass to a subsequence $\{ n_{j} \}\subseteq S$ such that the limit
$$\lim\limits_{j\rightarrow\infty}| \widehat{\mu}_{n_{j}}(4)|$$ exists. For example,
if $\{ n_{j}\} \subseteq S$ is a subsequence such that $\mu_{n_{j}}\Rightarrow\mu$ for some probability measure $\mu$ on $\Sc^{1}$,
then \eqref{eq:var leading KKW} reads (under the usual extra-assumption $\Nc_{n_{j}}\rightarrow\infty$)
\begin{equation}\label{eq:var leading KKW2}
\var(\Lc_{n_j})\sim c({ \mu}) \frac{E_{n_j}}{\Nc_{n_j}^2}
\end{equation}
with $$c(\mu) = \frac{1+\widehat{\mu}(4)^2}{512},$$
where, here and for the rest of the paper, we write $a_n\sim b_n$ to indicate that the two positive sequences $\{a_n\}$ and $\{b_n\}$ are such that $a_n/b_n \rightarrow 1$, as $n\to\infty$. Here, the set of the possible values for the $4$th Fourier coefficient $\widehat{\mu}(4)$ attains the whole interval $[-1,1]$ (see \cite{KKW,KW}). This implies in particular that the possible values of the asymptotic
constant $c(\mu)$ attain the whole interval
$\left[\frac{1}{512},\frac{1}{256}\right];$ the above is a complete classification of the asymptotic behaviour of $\var(\Lc_{n})$.

\subsection{Statement of the main results: asymptotic distribution of the nodal length}

Our main goal is the study of the fine asymptotic behaviour, as $\mathcal N_{n}\to \infty$, of the distributions of the sequence of normalised random variables
\begin{equation}\label{e:culp}
\widetilde{\Lc}_{n} :=  \frac{\mathcal{L}_{n} - \E[\mathcal{ L}_{n}]}{\sqrt{\Var(\mathcal{L}_{n} )}},
\quad  n\in S,
\end{equation}
(this is equivalent to studying $\widetilde{\Lc}_{n_j}$ along subsequences $\{n_{j}\}_{j\ge 1}\subseteq S$ satisfying $\mathcal N_{n_{j}}\rightarrow\infty$; note that it is possible to choose a full density subsequence in $S$ as above).
Since the variance \eqref{eq:var leading KKW} {diverges to infinity}, it seems
reasonable to expect a central limit result, that is, that the sequence $\{\widetilde{\Lc}_{n}\}$
converges in distribution to a standard Gaussian random variable. Our findings not only contradict this (somewhat naive) prediction,
but also classify all the weak-$\ast$ adherent points of the probability distributions associated with the collection of random variables $\left\{\widetilde{\Lc}_{n} : n\in S\right\}$
(where the adherent points are in the sense of weak-$\ast$ convergence of probability measures). In particular, we will show that such a set of weak-$\ast$ adherent points coincides with the collection of probability distributions associated with a family of linear combinations of two independent squared Gaussian random variables; these linear combinations are parameterized by the adherent points of the sequence $\left\{\left|\widehat{\mu_{n}}(4)\right|\right\}$ of real non-negative numbers
$\le 1$. This will show the remarkable fact that the angular distribution
of $\Lambda_{n}$ (or, more specifically, the $4$th Fourier coefficient of $\mu_{n}$)
does not only prescribe the leading term of the nodal length variance $\var(\mathcal{L}_{n})$, but, in addition, it prescribes the asymptotic distribution of $\widetilde{\Lc}_{n}$.

\vspace{5mm}

To state our results formally, we will need some more notation. For $\eta\in [0,1]$, let $\mathcal{M}_\eta$ be the random variable
\begin{equation}\label{e:r}
\mathcal{M}_\eta := \frac{1}{2\sqrt{1+\eta^2}} (2 - (1+\eta) X_1^2-(1-\eta) X_2^2),
\end{equation}
where $X=(X_{1},X_{2})$ are independent standard Gaussians. Note that
for $\eta_{1}\ne \eta_{2}$ the distributions of $\mathcal{M}_{\eta_1}$ and $ \mathcal{M}_{\eta_2}$
are genuinely different; this follows for example from the observation that the support of the distribution of $\Mc_\eta$ is $$\left ( -\infty,\frac{1}{\sqrt{1+\eta^2}}\right ].$$

Our first main result establishes a limiting law for the nodal length distribution for subsequences
$\{n_{j}\}_{j\ge 1}\subseteq S$ provided that the numerical sequence $$\big\{\big|\widehat{\mu}_{n_j}(4)\big| : j\geq 1 \big\}$$
of non-negative numbers is convergent. As it was mentioned above, for some full density subsequence $\{n_{j}\}_{j\ge 1}\subseteq S$ the corresponding lattice points $\Lambda_{n_{j}}$ are asymptotically equidistributed \eqref{eq:mun equidist},
so that for this subsequence, in particular, $$\widehat{\mu}_{n_j}(4) \rightarrow 0.$$ More generally,
if for some subsequence $\{n_{j}\}_{j\ge 1}\subseteq S$ the angular distribution of the corresponding
lattice points converges to $\mu$,
i.e. $\mu_{n_{j}}\Rightarrow \mu$, where $\mu$ is some probability measure on $\Sc^{1}$, then
$$\widehat{\mu}_{n_j}(4) \rightarrow \widehat{\mu}(4).$$ From now on, we use the symbol $\stackrel{\rm d}{\longrightarrow}$ to denote convergence in distribution of random variables; similarly, we will write $X \stackrel{\rm d}{=} Y$ to indicate that the random variables $X$ and $Y$ have the same distribution. Observe that a sequence of random variables converges in distribution if and only if
the corresponding sequence of probability laws is weak-$*$ convergent. We shall however use the sentence ``convergence in distribution" (resp. ``weak-$*$ convergence")  for random variables (resp.  for probability measures).

\begin{theorem}
\label{thm:lim dist sep}

Let $\{ n_{j} \}\subseteq S$ be a subsequence of $S$ satisfying $\mathcal N_{n_{j}}\rightarrow\infty$, such that the sequence
$\big\{\big|\widehat{\mu_{n_j}}(4)\big| : j\geq 1 \big\}$ of non-negative numbers converges, that is:
$$|\widehat{\mu}_{n_j}(4)\big|\rightarrow \eta,$$ for some $\eta \in [0,1]$.
Then
\begin{equation}
\label{eq:Lctild->Meta}
\widetilde{\Lc}_{n_j}  \stackrel{\rm d}{\longrightarrow} \mathcal{M}_\eta,
\end{equation}
where $\mathcal{M}_\eta$ was defined in \eqref{e:r}.
\end{theorem}
Since ~\cite{KKW,KW} showed that the set of adherent points of $\{\widehat{\mu}_{n}(4)\}_{n\in S}$ is all of $[-1,1]$, the result above clearly implies that $\widetilde{\Lc}_{n}$ does not converge in distribution for
$\mathcal N_{n}\rightarrow\infty$; in particular, if the sequence $\{|\widehat{\mu}_{n_j}(4)|\}$ does not converge, then the set of probability distributions associated with the random variables $\{ \widetilde{\Lc}_{n_j}\}$ has at least two different adherent points in the topology of weak-$\ast$ convergence.
It would be desirable to formulate a uniform asymptotic result a la \eqref{eq:Lctild->Meta}
with no separation of the full sequence $S$ into subsequences according to the angular distribution of $\Lambda_{n}$ (still as $\mathcal N_{n}\rightarrow \infty$). This has two subtleties though.

%\begin{verbatim}
%The use of \mathbb{Q} is confusing with rational numbers???
%\end{verbatim}

First, since there is no convergence in distribution,
we need to couple the random variables on the same probability space and work with some
metric on the space of probability
measures; we
choose to work with the $L^{p}$-metrics, $p\in (0,2)$.
Second, as, given a number $n\in S$, there is no limiting value $\eta$ of
$\widehat{\mu}_{n}(4)$, for each $n\in S$ the candidate $\mathcal{M}_\eta$ for the limiting random variable
will bear $$\eta=\eta_{n} = |\widehat{\mu}_{n}(4)|$$ rather than its limiting value. The following result is the desired refinement of Theorem \ref{thm:lim dist sep}. Its proof (omitted) boils down to a standard adaptation of the proof of \cite[Theorem 11.7.1]{D}, which is in turn an extension of the well-known {\it Skorohod representation Theorem} (see \cite[Theorem 11.7.2]{D}) to the framework of double sequences of probability measures.

\begin{theorem}
\label{c:coupling}
On some auxiliary probability space $(A, \mathscr{A},\widetilde \P)$ for every $n\in S$ there exists a coupling of the random variables $\widetilde{\Lc}_{n}$ and $\Mc_{|\widehat{\mu}_{n}(4)|}$
such that, as $\mathcal N_{n}\rightarrow\infty$,
\begin{equation}\label{e:bara1}
\E_{\widetilde \P} \left[\left|\widetilde{\Lc}_{n}-  \Mc_{|\widehat{\mu}_{n}(4)|}\right|^p\right]\rightarrow 0,
\end{equation}
for every $p \in (0,2)$, and
\begin{equation}\label{e:bara2}
\widetilde{\Lc}_{n}-  \Mc_{|\widehat{\mu}_{n}(4)|} \to 0, \quad \mbox{a.s.}-\widetilde \P.
\end{equation}
\end{theorem}
Relation \eqref{e:bara2} is equivalent to saying that, for every sequence $\{n_j\}\subseteq S$ such that $\Nc_{n_j}\to \infty$, $\widetilde \P\big(\widetilde{\Lc}_{n_j}-  \Mc_{|\widehat{\mu}_{n_j}(4)|} \to 0\big) =1$.  The fact that Theorem \ref{c:coupling} is actually a strenghtening of Theorem \ref{thm:lim dist sep} follows from the observation that, under the most natural coupling of the family of variables $\{\Mc_{\eta}\}_{\eta\in [0,1]}$ we have
\begin{equation*}
\E\left[\left| \Mc_{\eta_{1}} - \Mc_{\eta_{2}}   \right|\right] \le c |\eta_{1}-\eta_{2} |,
\end{equation*}
for all $\eta_{1},\eta_{2}\in [0,1]$ (with $c>0$ an absolute constant).
In fact, by the triangle inequality and an immediate computation, Theorem \ref{c:coupling} implies the stronger $L^{p}$-convergence, $p \in (0,2)$, to suitably coupled $\Mc_{\eta}$ in \eqref{eq:Lctild->Meta}.

\subsection{On the proofs of the main results}\label{on the proofs}

In Proposition \ref{teoexp} we compute the \emph{Wiener-It\^o chaos expansion} for
the nodal length $\mathcal L_n$ (\ref{e:length}), i.e.
a series converging in $L^2(\P)$ of the form
\begin{equation}
\label{eq:len chaos exp}
\mathcal L_n = \sum_{q=0}^{\infty} \text{proj}(\mathcal L_n | C_q)= \sum_{q=0}^{\infty} \mathcal L_n[q].
\end{equation}
Here $C_q$, $q=0,1,\dots$ are the
so-called {\it Wiener chaoses} (see \S \ref{ss:berryintro}), namely the orthogonal components of the $L^2$-space of those random variables that are functionals of some Gaussian white noise on $\mathbb T$ --
while $\mathcal L_n[q]:=\text{proj}(\mathcal L_n | C_q)$ denotes the orthogonal projection of $\mathcal L_n$ onto
the $q$-th chaos.

The decomposition \eqref{eq:len chaos exp} is of independent interest, and entails in particular the vanishing of
all the odd-order chaotic components and the term of order two, i.e.
$\mathcal L_n[q]=0$ if
$q=2m+1, m=0,1,\dots$ or $q=2$.
The precise analysis of the asymptotic behavior of the fourth-order projection in
Proposition \ref{prop:main result on proj4} will allow us
to show that its variance %$\var(\text{proj}(\mathcal L_n | C_4))$
is asymptotic to the total variance of
the nodal length (see Proposition \ref{eq:Lproj4 replace L});
since the different components are orthogonal by construction,
this will imply that all the projections other than the one on the fourth chaos are negligible.
We notice that it is relatively easy to show that the contribution to the nodal length
variance of each of the chaotic projections of order $q\neq 4$ is negligible.
It is in principle also possible to directly bound the total contribution to the variance of the sum of all these projections, thus establishing relation \eqref{eq:var leading KKW} independently. However, this task seems to be technically demanding, and would make our argument significantly longer. Since the asymptotic result \eqref{eq:var leading KKW} is already available from ~\cite{KKW}, we do not pursue such a strategy in the present manuscript.

As a consequence, to study the asymptotic behavior of $\mathcal L_n$ it will be sufficient to focus
on the above-mentioned fourth-order component; Proposition \ref{prop:main result on proj4}
shows that %there is \emph{no} convergence in distribution, and
along subsequences $\lbrace n_j \rbrace$ satisfying the same hypothesis
as in Theorem \ref{thm:lim dist sep},
we have
$$
\frac{\mathcal L_{n_j}[4]}{\sqrt{\Var(\mathcal L_{n_j}[4])}}
\stackrel{\rm d}{\longrightarrow} \mathcal{M}_\eta,
$$
where $\mathcal M_\eta$ is as in (\ref{e:r}).

We are then able to prove Theorem \ref{thm:lim dist sep} thanks to
Proposition \ref{eq:Lproj4 replace L} and Proposition \ref{prop:main result on proj4}.

\subsection{Plan of the paper}

In \S \ref{ss:berryintro} we recall Wiener-It\^o chaotic expansions,
which we then exploit throughout the whole paper to prove the main results,
given in \S \ref{secProof}; \S \ref{expan} is devoted to the proof of the chaotic expansion for the nodal length (Proposition \ref{teoexp}), whereas in \S \ref{dimProp}
we prove Proposition \ref{prop:main result on proj4} and Proposition
\ref{eq:Lproj4 replace L}. Finally, in \S \ref{lemmaFor} we collect the technical proofs of auxiliary lemmas for the results given in \S \ref{dimProp}.

\subsection{Acknowledgements}

The research leading to these results has received funding from the
European Research Council under the European Union's Seventh
Framework Programme (FP7/2007-2013) / ERC grant agreements n$^{\text{o}}$ 277742
\emph{Pascal} (Domenico Marinucci and Maurizia Rossi) and
n$^{\text{o}}$ 335141 \emph{Nodal} (Igor Wigman), and
by the grant F1R-MTH-PUL-15STAR (STARS)
at Luxembourg University (Giovanni Peccati and Maurizia Rossi).
We are grateful to Zeev Rudnick and Peter Sarnak for many insightful conversations and to two anonymous referees for valuable suggestions and insightful remarks.

\section{Proofs of the main results}
The proofs of our results rely on a pervasive use of \emph{Wiener-It\^o chaotic expansions} for non-linear functionals of Gaussian fields; this notion is presented below in a form that is adapted to the random functions considered in the present paper (see e.g. \cite{N-P, P-T}
for an exhaustive discussion).
\subsection{Wiener Chaos}
%and Berry's cancellation phenomenon
\label{ss:berryintro}

 Denote by $\{H_k\}_{k\ge 0}$ the usual Hermite polynomials on $\mathbb{R}$. These are defined recursively as follows: $H_0 \equiv 1$, and, for $k\geq 1$,

 $$H_{k}(t) = tH_{k-1}(t) - H'_{k-1}(t).$$
Recall that $\mathbb{H} := \{[k!]^{-1/2} H_k : k\ge 0\}$ constitutes a complete orthonormal system in $$L^2(\mathbb{R}, \mathscr{B}(\mathbb{R}), \gamma(t)dt) :=L^2(\gamma),$$ where $\gamma(t) = (2\pi)^{-1/2}e^{-t^2/2}$ is the standard Gaussian density on the real line.

\medskip

The arithmetic random waves \eqref{defrf} considered in this work are a by-product of a family of complex-valued Gaussian random variables $\{a_\lambda : \lambda\in \mathbb{Z}^2\}$, defined on some probability space $(\Omega, \mathscr{F}, \mathbb{P})$ and satisfying the following properties: {\bf (a)} every $a_\lambda$ has the form $x_\lambda+iy_\lambda$, where $x_\lambda$ and $y_\lambda$ are two independent real-valued Gaussian random variables with mean zero and variance $1/2$; {\bf (b)} $a_\lambda$ and $a_\tau$ are stochastically independent whenever $\lambda \notin\{ \tau, -\tau\}$, and {\bf (c)} $a_\lambda = \overline{a_{-\lambda}}$. Define the space ${\bf A}$ to be the closure in $L^2(\mathbb{P})$ of all real finite linear combinations of random variables $\xi$ of the form $$\xi = z \, a_\lambda + \overline{z} \, a_{-\lambda},$$ where $\lambda\in \mathbb{Z}^2$ and $z\in \mathbb{C}$. The space ${\bf A}$ is a real centered Gaussian Hilbert subspace
of $L^2(\mathbb{P})$.

\begin{defn}\label{d:chaos}{\rm For an integer $q\ge 0$ the $q$-th {\it Wiener chaos} associated with ${\bf A}$, written $C_q$, is the closure in $L^2(\mathbb{P})$ of all real finite linear combinations of random variables of the form
$$
H_{p_1}(\xi_1)\cdot H_{p_2}(\xi_2)\cdots H_{p_k}(\xi_k)
$$
for $k\ge 1$, where the integers $p_1,...,p_k \geq 0$ satisfy $p_1+\cdots+p_k = q$, and $(\xi_1,...,\xi_k)$ is a standard real Gaussian vector extracted
from ${\bf A}$ (note that, in particular, $C_0 = \mathbb{R}$).}
\end{defn}

Using the orthonormality and completeness of $\mathbb{H}$ in $L^2(\gamma)$, together with a standard monotone class argument (see e.g. \cite[Theorem 2.2.4]{N-P}), it is not difficult to show that $C_q \,\bot\, C_m$ (where the orthogonality holds in the sense of $L^2(\mathbb{P})$) for every $q\neq m$, and moreover
\begin{equation*}
L^2(\Omega, \sigma({\bf A}), \mathbb{P}) = \bigoplus_{q=0}^\infty C_q;
\end{equation*}
that is, every real-valued functional $F$ of ${\bf A}$ can be (uniquely) represented in the form
\begin{equation}\label{e:chaos2}
F = \sum_{q=0}^\infty {\rm proj}(F \, | \, C_q)=\sum_{q=0}^\infty F[q],
\end{equation}
where as before $F[q]:={\rm proj}(F \, | \, C_q)$ stands for the the projection onto $C_q$, and the series converges in $L^2(\mathbb{P})$. Plainly, $F[0]={\rm proj}(F \, | \, C_0) = \E [F]$.

\smallskip

A straightforward differentiation of the definition \eqref{defrf} of $T_n$ yields, for $j=1,2$
\begin{equation}\label{e:partial}
\partial_j T_n(x) = \frac{2\pi i}{\sqrt{\mathcal{N}_n} }\sum_{(\lambda_1,\lambda_2)\in \Lambda_n} \lambda_j a_\lambda e_\lambda(x),
\end{equation}
(here $\partial_j = \frac{\partial}{\partial x_j}$).
Hence the random fields $T_{n},\partial_{1} T_n,\partial_{2} T_n$ viewed as collections of
Gaussian random variables
indexed by $x\in\Tb$ are all lying in ${\bf A}$, i.e. for every $x\in\Tb$ we have
\begin{equation*}
T_{n}(x),\, \partial_{1}T_{n}(x), \, \partial_{2}T_{n}(x) \in \bf A.
\end{equation*}

\subsection{Proof of Theorem \ref{thm:lim dist sep}}\label{secProof}

We apply the Wiener chaos decomposition \eqref{e:chaos2} on the nodal length
\begin{equation}\label{eq:len chaos decomp}
\mathcal{L}_{n} = \sum_{q=0}^\infty \mathcal{L}_{n}[q],
\end{equation}
in $L^2(\mathbb{P})$.
The following proposition is a reformulation
of Theorem \ref{thm:lim dist sep} with the projection $\mathcal{L}_{n}[4]$ of the nodal length
$\mathcal{L}_{n}$ onto the $4$th order chaos replacing $\mathcal{L}_{n}$ and
it will be proven in \S \ref{proof4}.

\begin{proposition}
\label{prop:main result on proj4}

Let $\{ n_{j} \}\subseteq S$ be a subsequence of $S$ satisfying $\mathcal N_{n_{j}}\rightarrow\infty$, such that the sequence
$\big\{\big|\widehat{\mu_{n_j}}(4)\big| : j\geq 1 \big\}$ of non-negative numbers converges, that is,
$$|\widehat{\mu_{n_j}}(4)\big|\rightarrow \eta,$$ for some $\eta \in [0,1]$.
Then, the corresponding sequences of random variables converges in distribution to $\mathcal{M}_\eta$ as defined in \eqref{e:r}, that is,
\begin{equation}
\label{eq:Lctild->Meta2}
\frac{\mathcal{L}_{n_j}[4] }{\sqrt{\var(\mathcal{L}_{n_j}[4] ) }} \stackrel{\rm d}{\longrightarrow} \mathcal{M}_\eta.
\end{equation}
Moreover,
\begin{equation}\label{perlavar}
\Var\left (  \mathcal{L}_{n_j}[4]    \right )\sim \frac{1+\eta^2}{512}\frac{E_{n_j}}{\mathcal N_{n_j}^2}.
\end{equation}
\end{proposition}
%
%Note that if the sequence $\{n_j\}$ is such that the (bounded) sequence $|\widehat{\mu}_{n_j}(4)|$ is not converging,
%then $\frac{{\rm proj}(\mathcal{L}_{n_j}\, | C_{4} )}{\sqrt{\var({\rm proj}(\mathcal{L}_{n_j} \, | C_{4} ))} }$ is not converging in distribution, since in this case the  probability distribution set
%of $\frac{{\rm proj}(\mathcal{L}_{n_j}\, | C_{4} )}{\sqrt{\var({\rm proj}(\mathcal{L}_{n_j} \, | C_{4} ))} }$, $j\ge 1$, has necessarily at least two distinct adherent points.

The next proposition, whose proof is given in \S \ref{proof4}, entails that the fourth-order chaotic component gives the leading term in the expansion,
i.e. its behaviour asymptotically dominates the nodal length on the torus.

\begin{proposition}
\label{eq:Lproj4 replace L}
For every $\{n_j : j\geq 1\}\subseteq S$ subsequence of $S$ such that
$\lim_{j\to\infty }{\mathcal N}_{n_j}  = \infty$ and the sequence
$\big\{\big|\widehat{\mu_{n_j}}(4)\big| : j\geq 1 \big\}$ of non-negative numbers converges,
\begin{equation}\label{vaar2}
\var\left(\mathcal{L}_{n_j}- \mathcal{L}_{n_j}[4]\right)=o\left( \frac{E_{n_j}}{\Nc_{n_j}^2}\right).
\end{equation}
Equivalently, under the above assumptions we have that
\begin{equation}\label{vaar}
\var\left(\mathcal{L}_{n_j}\right)\sim \var\left(\mathcal{L}_{n_j}[4] \right).
\end{equation}
\end{proposition}
\noindent \begin{proof}[Proof of Theorem \ref{thm:lim dist sep} assuming Proposition \ref{prop:main result on proj4} and Proposition \ref{eq:Lproj4 replace L}] The chaotic expansion \eqref{eq:len chaos decomp} and Proposition \ref{eq:Lproj4 replace L} entail that, as $j\to +\infty$,
$$
\widetilde{\mathcal L}_{n_j} = \widetilde{\mathcal{L}}_{n_j}[4] + o_\P(1),
$$
where $o_\P(1)$ denotes a sequence of random variables converging to zero in probability. Actually, by linearity
we have
\begin{equation}\label{facile}
\widetilde{\mathcal{L}}_{n_j}[4]=\frac{\mathcal{L}_{n_j}[4]}{\sqrt{\Var(\mathcal L_{n_j})}}.
\end{equation}
It hence follows that $\widetilde{\mathcal L}_{n_j}$ and the random variable $\widetilde{\mathcal{L}}_{n_j}[4]$ have the same asymptotic distribution.
Proposition \ref{eq:Lproj4 replace L} together with \paref{eq:Lctild->Meta2} and \paref{facile} allow to conclude the proof, i.e. they immediately imply \paref{eq:Lctild->Meta}.
\end{proof}
\begin{remark}[On the length of $u$-level curves]\label{remark non-nodal}\rm

For $u\in \mathbb R$, let us consider the total length of $u$-level curves for arithmetic random waves, i.e. the sequence of random variables $\lbrace \mathcal L_{n;u}\rbrace_{n\in S}$ defined as
$$ \mathcal L_{n;u}:= \text{length}(T_n^{-1}\lbrace u\rbrace).$$
Of course, $\mathcal L_{n;0}=\mathcal L_n$. The behaviour of $\lbrace \mathcal L_{n;u}\rbrace_{n\in S}$ for $u \neq 0$ exhibits rather different characteristics than for the nodal case.
Indeed, following a slightly modified version of the arguments we develop here (based on Green's formula and the properties of Laplacian eigenfunctions - see i.e. \cite[\S 7.3 and p.134]{rossiphd}), it can be shown that the second order chaotic projection of $\mathcal L_{n;u}$ is given by
\begin{equation}\label{2chaosu}
\mathcal L_{n;u}[2]= \sqrt{\frac{E_n}{2}} \sqrt{\frac{\pi}{8}} \phi(u)u^2 \frac{1}{\mathcal N_n} \sum_{\lambda\in \Lambda_n} (|a_\lambda|^2-1),
\end{equation}
where $\phi$ denotes the standard Gaussian density; note that \paref{2chaosu} confirms $\mathcal L_n[2]=0$ in the nodal case (see \S \ref{on the proofs} and Proposition \ref{teoexp} (a)).

A few comments are in order. Let us first notice that the asymptotic variance of $\mathcal L_{n;u}[2]$ satisfies, as $n\to +\infty$ such that $\mathcal N_n\to +\infty$,
\begin{equation}\label{variancechaos2}
\Var(\mathcal L_{n;u}[2]) \sim \frac{\e^{-u^2}u^4}{8}\frac{E_n}{\mathcal N_n}.
\end{equation}
The variance of the length of $u$-level curves for $u \ne 0$ can be derived exploiting the same computations as in \cite{KKW}, and it is then possible to check the following asymptotic equivalence: for $u\ne 0$, as $\mathcal N_n\to +\infty$,
\begin{equation*}
\Var(\mathcal L_{n;u}) \sim \Var(\mathcal L_{n;u}[2]).
\end{equation*}
Hence, the variance of the length of non-zero level curves has a larger asymptotic order of magnitude than in the nodal case (compare (\ref{variancechaos2}) - (\ref{perlavar})); indeed, the former is dominated by the term corresponding to the second-order chaos, rather than the fourth. At $u=0$, the second-order chaos component of the length of $u$-level curves vanishes exactly, and thus the variance has a lower asymptotic magnitude, consistently with the so-called Berry's cancellation phenomenon  \cite{Berry 2002, wig, wigsurvey}.  Also, because the second-order chaos term  \paref{2chaosu} is proportional to a simple sum of independent, identically distributed, finite-variance centred random variables (discounting repetitions coming from the symmetric structure of $\Lambda_n$), it is trivial to show that it exhibits limiting Gaussian behaviour, in marked contrast with the non-universal and non-Central Limit Theorem emerging in the nodal case.

\end{remark}

\begin{remark}[On local statistics]\label{remark local}\rm

 Our method can be applied, in principle, to prove limit theorems for the nodal length within a proper subregion of the torus too. While the derivation of the $L^2$-expansion into Hermite polynomials does not require any new ideas or techniques, some of the variance computations, and consequently the limiting distribution, will be affected.
Note indeed that the main results of the present paper are obtained by exploiting
 some exact cancellations which are taking place when evaluating integrals of the eigenfunctions on the full torus. We hence leave these generalizations as a topic for future research.
\end{remark}

\section{Chaotic expansion of ${\Lc}_{n}$
}
\label{expan}

In order to prove Proposition \ref{prop:main result on proj4} and Proposition \ref{eq:Lproj4 replace L}
we need to compute the Wiener-It\^o chaotic expansion \eqref{eq:len chaos decomp} of the random variable $\mathcal L_n$; we refer to \cite{KL} for analogous computations involving the length of level curves in the case of two-dimensional Gaussian fields on the Euclidean plane.

\subsection{Statement}

Let us introduce some more notation to properly state the main result of this section. The nodal length (\ref{e:length}) can be formally written as
\begin{equation}\label{formalLength}
\mathcal L_n = \int_{\mathbb T} \delta_0(T_n(\theta))\|\nabla T_n(\theta)\|\,d\theta,
\end{equation}
where $\delta_0$ denotes the Dirac delta function and $\|\cdot \|$ the
Euclidean norm in $\R^2$ (see \cite[Lemma 3.1]{RudWig} and \S \ref{prelRes}).

We shall often use the following easy result from \cite{RudWig}:

\begin{lemma}[\cite{RudWig}, (4.1)]\label{lemmavar}
For $j=1,2$ we have that
$$\displaylines{
\Var[\partial_j T_n(x)] = \frac{4\pi^2}{\mathcal N_n}
\sum_{\lambda\in \Lambda_n} \lambda_j^2 = 4\pi^2 \frac{n}{2},
}$$
where the derivatives $\partial_j T_n(x)$ are as in \eqref{e:partial}.
\end{lemma}
Accordingly, for $x=(x_1, x_2)\in \mathbb T$ and $j=1,2$, we will denote by $\partial_j \widetilde T_n(x)$ the normalized derivative
\begin{equation}\label{e:norma}
\partial_j \widetilde T_n(x) := \frac{1}{2\pi} \sqrt{\frac{2}{n}} \frac{\partial}{\partial x_j}  T_n(x) = \sqrt{\frac{2}{n}}\frac{ i}{\sqrt{\mathcal N_n}}\sum_{ \lambda\in \Lambda_n}\lambda_j\,
a_{\lambda}e_\lambda(x).
\end{equation}
In view of convention \eqref{e:norma}, we formally rewrite \eqref{formalLength} as
\begin{equation*}\label{formalLength2}
\mathcal L_n = \sqrt{\frac{4\pi^2n}{2}}\int_{\mathbb T}
\delta_0(T_n(x)) \sqrt{\partial_1 \widetilde
T_n(x)^2+\partial_2 \widetilde T_n(x)^2}\,dx.
\end{equation*}
We also introduce two collections of coefficients
$\{\alpha_{2n,2m} : n,m\geq 1\}$ and $\{\beta_{2l} : l\geq 0\}$, that are related to the (formal) Hermite expansions of the norm $\| \cdot
\|$ in $\R^2$ and the Dirac mass $ \delta_0(\cdot)$ respectively.
These are given by
\begin{equation}\label{e:beta}
\beta_{2l}:= \frac{1}{\sqrt{2\pi}}H_{2l}(0),
\end{equation}
where $H_{2l}$ denotes the $2l$-th Hermite polynomial, and
\begin{equation}\label{e:alpha}
\alpha_{2n,2m}=\sqrt{\frac{\pi}{2}}\frac{(2n)!(2m)!}{n!
m!}\frac{1}{2^{n+m}} p_{n+m}\left (\frac14 \right),
\end{equation}
where for $N=0, 1, 2, \dots $ and $x\in \R$
\begin{equation*}
\displaylines{ p_{N}(x) :=\sum_{j=0}^{N}(-1)^{j}\cdot(-1)^{N}{N
\choose j}\ \ \frac{(2j+1)!}{(j!)^2} x^j, }
\end{equation*}
$\frac{(2j+1)!}{(j!)^2}$ being the so-called {\it swinging factorial}
restricted to odd indices.

We are now ready to state the main result of this section. It illustrates the cancellations that occur for the components of the chaotic expansion \paref{eq:len chaos decomp} of $\mathcal L_n$ (precisely, odd terms and the second-order one). Consistent to Proposition \ref{eq:Lproj4 replace L},
computing the fourth-order component only is sufficient to establish the asymptotic behavior of the nodal length.
However, we believe that the complete expansion is of clear independent interest; for instance, (a) it gives the basic building block
to extend our results to other random fields on the torus and (b) it sheds some light on the
Berry's cancellation phenomenon \cite{Berry 2002, wig, wigsurvey}, as discussed also earlier in Remark \ref{remark non-nodal}.

More precisely, as far as point (b) is concerned, we note that the nodal length $\mathcal L_\ell$ of Gaussian Laplace eigenfunctions $T_\ell$, $\ell\in \N$,
on the two-dimensional sphere have the same qualitative behavior. Indeed, on one hand in the chaotic expansion of $\mathcal L_\ell$, the odd terms and the second chaotic projection vanish and the fourth-order component exhibits the same asymptotic variance as the full nodal length (see \cite{rossiphd}). On the other hand, it is also shown in \cite{rossiphd} that the second chaotic projection in the Wiener-It\^o expansion of the length of level curves $T_\ell^{-1}(u)$, $u\in \mathbb R$ vanishes if and only if $u=0$. These results explain why the asymptotic variance of the length of level curves is consistent to the natural scaling, except for the nodal case \cite{wig,wigsurvey}. Finally, we note that an analogous cancellation phenomenon occurs for the excursion area and the Euler-Poincar\'e characteristic of excursion sets for spherical eigenfunctions, see \cite{MaWi1,mrossi,cmw2015}.

\begin{proposition}[\bf Chaotic expansion of $\Lc_n$]
\label{teoexp}
\
{\rm (a)} For $q=2$ or $q=2m+1$ odd ($m\ge 1$),
\begin{equation*}\label{point1}
\Lc_n[q] \equiv 0,
\end{equation*}
that is, the corresponding chaotic projection vanishes.

{\rm (b)} For $q\geq 2$
\begin{eqnarray}\label{e:pp}
\nonumber &&\Lc_n[2q]\\
&&= \sqrt{\frac{4\pi^2n}{2}}\sum_{u=0}^{q}\sum_{k=0}^{u}
\frac{\alpha _{2k,2u-2k}\beta _{2q-2u}
}{(2k)!(2u-2k)!(2q-2u)!} \times\\
&&\hspace{4cm}  \times \int_{\mathbb T}\!\! H_{2q-2u}(T_n(x))
H_{2k}(\partial_1 \widetilde T_n(x))H_{2u-2k}(\partial_2
\widetilde T_n(x))\,dx.\notag
\end{eqnarray}
%In particular,
%\begin{equation*}
%{\rm proj}(\Lc_n\, | \, C_{4}) = ....
%\end{equation*}
Consolidating the above, the Wiener-It\^o chaotic expansion of $\Lc_n$ is
\begin{eqnarray*}%\label{chaosexp}
\Lc_n = \E \Lc_n + \sqrt{\frac{4\pi^2n}{2}}\sum_{q=2}^{+\infty}\sum_{u=0}^{q}\sum_{k=0}^{u}
\frac{\alpha _{2k,2u-2k}\beta _{2q-2u}
}{(2k)!(2u-2k)!(2q-2u)!}\times\\
\nonumber
\times \int_{\mathbb T}H_{2q-2u}(T_n(x))
H_{2k}(\partial_1 \widetilde T_n(x))H_{2u-2k}(\partial_2
\widetilde T_n(x))\,dx,
\end{eqnarray*}
in $L^2(\P)$.
\end{proposition}

\subsection{Proof of Proposition \ref{teoexp}}\label{proofVar}

Let us start with an approximating result in $L^2(\P)$ for the nodal length $\mathcal L_n$.

\subsubsection{Approximating the nodal length}\label{prelRes}

%\begin{verbatim}
%Rewrite the following top down - explain why we need this.
%More references and more explanations (in particular Cauchy, Jensen, Monotone Convergence etc).
%\end{verbatim}
%For Igor: No longer needed, now the proof is very short

Consider the family of random variables $\lbrace \mathcal L_n^\varepsilon,
\varepsilon > 0\rbrace$ defined as
\begin{equation}\label{napprox}
\mathcal L_n^\varepsilon = \frac{1}{2\varepsilon}\int_{\mathbb T}
1_{[-\varepsilon, \varepsilon]}(T_n(x)) \| \nabla T_n(x)\|dx,
\end{equation}
where $1_{[-\varepsilon, \varepsilon]}$ is the indicator function of
the interval $[-\varepsilon, \varepsilon]$, and $\|\cdot \|$ is the standard Euclidean norm in $\R^2$.

In view of the convention \eqref{e:norma} we rewrite \paref{napprox} as
\begin{equation*}\label{formal2}
\mathcal L_n^\varepsilon = \sqrt{\frac{4\pi^2n}{2}}\frac{1}{2\varepsilon}\int_{\mathbb T}
1_{[-\varepsilon, \varepsilon]}(T_n(x)) \sqrt{\partial_1 \widetilde
T_n(x)^2+\partial_2 \widetilde T_n(x)^2}\,dx.
\end{equation*}
In \cite[Lemma 3.1]{RudWig} it was shown that, a.s.
\begin{equation}\label{as-conv}
\mathcal L_n = \lim_{\varepsilon \to 0} \mathcal L^\varepsilon_n,
\end{equation}
(a rigorous manifistation of \eqref{formalLength}), and moreover, by
\cite[Lemma 3.2]{RudWig}, $\mathcal L_n^\varepsilon$ is uniformly bounded, that is:
\begin{equation}
\label{eq:leneps unif bounded}
\Lc_n^{\varepsilon} \leq 12 \sqrt{E_n}.
\end{equation}
Applying the Dominated Convergence Theorem to \eqref{as-conv} while bearing in mind
the uniform bound \eqref{eq:leneps unif bounded} implies that the convergence in \eqref{as-conv} is
in $L^2(\P)$, i.e. the following result:

\begin{lemma}\label{approx}
For every $n\in S$, we have
$$\lim_{\varepsilon \to 0} \E[|\Lc_n^{\varepsilon}-\mathcal L_n|^2]=0.$$
\end{lemma}

\subsubsection{Proof of Proposition \ref{teoexp}: technical computations}\label{proofVar2}

In view of Lemma \ref{approx}, we first compute the chaotic expansion
of $\mathcal L_n^\varepsilon$ and then deduce Proposition \ref{teoexp} by letting $\varepsilon\to 0$. Let us start by expanding the function $\frac{1}{2\varepsilon}{1}_{[-\varepsilon,\varepsilon]}(\cdot)$
into Hermite polynomials, as defined in \S \ref{ss:berryintro}.

\begin{lemma}\label{indicator function}
The following decomposition holds in $L^2(\gamma)$ (where, as before, $\gamma$ is the standard Gaussian density on $\R$):
\begin{equation*}
\frac{1}{2\varepsilon}{1}_{[-\varepsilon ,\varepsilon ]}(\cdot
)=\sum_{l=0}^{+\infty }\frac{1}{l!}\beta_l^{\varepsilon}\,H_{l}(\cdot ),
\end{equation*}
where, for $l \geq 1$
\[
\beta_l^{\varepsilon}=-\frac{1}{2\varepsilon}
\gamma \left (\varepsilon \right) \left (H_{l-1} \left (\varepsilon \right)-
H_{l-1} \left (-\varepsilon \right) \right ),
\]
while for $l=0$
\[
\beta_0^{\varepsilon} = \frac{1}{2\varepsilon} \int_{-\varepsilon}^{\varepsilon} \gamma(t)\,dt.
\]
Moreover, as $\varepsilon\to 0$,
$$
\beta_{l}^\varepsilon \to \beta_l,
$$
where for odd $\ell$, $\beta_l=0$ whereas $\beta_l$ coincides with \paref{e:beta} for even $\ell$.
\end{lemma}
\noindent\begin{proof}
Using the completeness and orthonormality of the set $\mathbb{H}$ in $L^2(\gamma)$, one has that $\beta_0^{\varepsilon} = \frac{1}{2\varepsilon} \int_{-\varepsilon}^{\varepsilon} \gamma(t)\,dt$, and, for $l\geq 1$,
$$\displaylines{
\beta_l^{\varepsilon} = \frac{1}{2\varepsilon} \int_{-\varepsilon}^{\varepsilon} \gamma(t) H_l(t)    \,dt=
\frac{1}{2\varepsilon} \int_{-\varepsilon}^{\varepsilon} \gamma(t) (-1)^l \gamma^{-1}(t) \frac{d^l}{dt^l} \gamma(t)  \,dt =\cr
=\frac{1}{2\varepsilon} (-1)^l \left (\frac{d^{l-1}}{dt^{l-1}}\gamma \left (\varepsilon \right)-
\frac{d^{l-1}}{dt^{l-1}}\gamma \left (-\varepsilon \right) \right )=-\frac{1}{2\varepsilon}
\gamma \left (\varepsilon \right) \left (H_{l-1} \left (\varepsilon \right)-
H_{l-1} \left (-\varepsilon \right) \right ).
}$$
Now, if $l$ is odd, then $H_{l-1}$ is an even function, and therefore $\beta_l^{\varepsilon}=0$: it follows that
\begin{equation*}
\frac{1}{2\varepsilon}{1}_{[-\varepsilon,\varepsilon ]}(\cdot
)=\beta_0^{\varepsilon} + \sum_{l=1}^{+\infty }\frac{1}{(2l)!}\left (-\frac{1}{\varepsilon}\gamma
\left (\varepsilon \right )H_{2l-1}\left (\varepsilon \right )\right )\,H_{2l}(\cdot ).
\end{equation*}%
Using the notation \eqref{e:beta}, we have that, for all $l\geq 0$,
\begin{equation}\label{e:sat}
\lim_{\varepsilon} \beta_{2l}^{\varepsilon} =- \frac{1}{\sqrt{2\pi }}(2l-1)!\frac{(-1)^{l-1}}{\left( l-1\right)!2^{l-1}} =\frac{1}{\sqrt{2\pi }} H_{2l}(0) = \beta_{2l}.
\end{equation}
\end{proof}

Note that the set $\{\beta_l : l = 0,1,2,... \}$ can be interpreted as the sequence of the coefficients appearing in the formal Hermite expansion of the Dirac mass $\delta_0$.

\smallskip

Now fix $x\in \mathbb{T}$, and recall that the coordinates of the vector $$\nabla \widetilde T_n(x)
 := (\partial_1 \widetilde T_n(x), \partial_2 \widetilde T_n(x)),
$$ are unit variance centered independent Gaussian random variables (see i.e., \cite{KKW}).
Now, since the random variable $\| \nabla \widetilde T_n(x) \|$ is square-integrable,
it can be expanded into an (infinite) series of Hermite polynomials, as detailed in the following statement.

\begin{lemma}\label{Euclidean norm}
For $(Z_1,Z_2)$ a standard Gaussian bivariate vector, we have the $L^2$-expansion
\begin{equation*}
\|(Z_1,Z_2)\| = \sum_{n=0}^{+\infty}
\sum_{m=0}^{n} \frac{\alpha_{2n,2n-2m}}{(2n)! (2n-2m)!} H_{2n}(Z_1) H_{2n-2m}(Z_2),
\end{equation*}
where the $\alpha_{2n,2n-2m}$ are as in \eqref{e:alpha}.
\end{lemma}
\noindent\begin{proof}
We may expand
\begin{equation*}
\|(Z_1,Z_2)\| = \sum_{u=0}^{+\infty}
\sum_{m=0}^{u} \frac{\alpha_{u,u-m}}{u! (u-m)!} H_u(Z_1) H_{u-m}(Z_2),
\end{equation*}
where
\begin{equation*}
\alpha_{n,n-m}=\frac{1}{2\pi} \int_{\R^2} \sqrt{y^2 + z^2} H_{n}(y) H_{n-m}(z)
\mathrm{e}^{-\frac{y^2+z^2}{2}}\,dy dz.
\end{equation*}
Our aim is to compute $\alpha_{n,n-m}$ as explicitly as possible. First of all, we observe that, if $n$ or $n-m$ is odd, then the above integral
vanishes (since the two mappings $z\mapsto \sqrt{y^2 + z^2}$ and $y\mapsto \sqrt{y^2 + z^2}$ are even). It follows therefore that
\begin{equation*}
\|(Z_1,Z_2)| = \sum_{n=0}^{+\infty}
\sum_{m=0}^{n} \frac{\alpha_{2n,2n-2m}}{(2n)! (2n-2m)!} H_{2n}(Z_1) H_{2n-2m}(Z_2).
\end{equation*}
We are therefore left with the task of showing that the integrals
\begin{equation*}  \label{coeff}
\alpha_{2n,2n-2m}=\frac{1}{2\pi} \int_{\R^2} \sqrt{y^2 + z^2} H_{2n}(y)
H_{2n-2m}(z) \mathrm{e}^{-\frac{y^2+z^2}{2}}\,dy dz,
\end{equation*}
where $n\ge 0$ and $m=0, \dots, n$, are given by \eqref{e:alpha}. One elegant way for dealing with this task is to use the following Hermite polynomial expansion (see e.g. \cite[Proposition 1.4.2]{N-P})
\begin{equation}\label{hermiteExp}
\mathrm{e}^{\lambda y - \frac{\lambda^2}{2}} = \sum_{a=0}^{+\infty} H_a(y)
\frac{\lambda^a}{a!}, \quad  \lambda \in \R.
\end{equation}
Let us consider the integral
\begin{equation*}
\frac{1}{2\pi} \int_{\mathbb R^2} \sqrt{y^2 + z^2} \e^{\lambda y - \frac{\lambda^2}{2}} \e^{\mu z - \frac{\mu^2}{2}}\e^{-\frac{y^2+z^2}{2}}\,dy dz
=\frac{1}{2\pi} \int_{\mathbb R^2} \sqrt{y^2 + z^2} \e^{-\frac{(y-\lambda)^2+(z-\mu)^2}{2}}\,dy dz.
\end{equation*}
This integral coincides with the expected value of the random variable $W:=\sqrt{%
Y^2 + Z^2}$ where $(Y,Z)$ is a vector of independent Gaussian random variables with variance one and mean $\lambda $ and $\mu$, respectively. Note that $W^2=Y^2 + Z^2$ has a non-central $\chi^2$-distribution: more precisely, $Y^2 + Z^2\sim \chi^2(2,\lambda^2 +
\mu^2)$. Its density $f_{W^2}$ (see e.g. \cite[(26.4.25)]{AS}) is given by
\begin{equation*}
f_{W^2}(t) = \sum_{j=0}^{+\infty} \mathrm{e}^{-(\lambda^2 + \mu^2)/2}\frac{%
((\lambda^2 + \mu^2)/2)^j }{j!}f_{2+2j}(t)\,{1}_{\{t> 0\}},
\end{equation*}
where $f_{2+2j}$ is the density function  of a $\chi^2_{2+2j}$-distributed random variable (see e.g. \cite[(26.4.1)]{AS}).
Therefore, the density $f_W$ of $W$ is $f_W(t) = f_{W^2}(t^2)\,2t$, i.e.
\begin{equation*}
f_W(t) = \sum_{j=0}^{+\infty} \mathrm{e}^{-(\lambda^2 + \mu^2)/2}\frac{%
((\lambda^2 + \mu^2)/2)^j }{j!}f_{2+2j}(t^2)\, 2t \,{1}_{\{t> 0\}}.
\end{equation*}
Hence, the expected
value of $W$ is
\begin{equation}\label{mediaW}
\E[W]=2\sum_{j=0}^{+\infty} \mathrm{e}^{-(\lambda^2 + \mu^2)/2}\frac{((\lambda^2 +
\mu^2)/2)^j }{j!}\int_{0}^{+\infty} f_{2+2j}(t^2)\, t^2\,dt.
\end{equation}
From the definition of $f_{2+2j}$ in \cite[(26.4.1)]{AS} we have
\begin{eqnarray}\nonumber
\int_{0}^{+\infty} f_{2+2j}(t^2)\, t^2\,dt &=&
\frac{1}{2^{1+j}\Gamma(1+j)} \int_{0}^{+\infty} t^{2j+2}\e^{-t^2/2}\,dt \\
\label{ciaobella}
&=& \frac{\prod_{i=1}^{1+j} (2i-1)\sqrt{\frac{\pi}{2}}}{2^{1+j}\Gamma(1+j)}.
\end{eqnarray}
Substituting  \paref{ciaobella} into \paref{mediaW} we have
\begin{eqnarray}\label{fico2}
 \E[W]
  =2\e^{-(\lambda^2 + \mu^2)/2}\sum_{j=0}^{+\infty} \frac{((\lambda^2
+ \mu^2)/2)^j }{j!}\frac{\prod_{i=1}^{1+j} (2i-1)\sqrt{\frac{\pi}{2}}}{2^{1+j}\Gamma(1+j)} =: F(\lambda, \mu).
\end{eqnarray}
Applying Newton's binomial formula to $((\lambda^2
+ \mu^2)/2)^j$, we may expand the function $F$ in \paref{fico2} as follows:
\begin{equation*}
\displaylines{
F(\lambda, \mu) =
2\sum_{a=0}^{+\infty} \frac{(-1)^a\lambda^{2a}}{2^a a!}
\sum_{b=0}^{+\infty} \frac{(-1)^b\mu^{2b}}{2^b b!}
\sum_{j=0}^{+\infty} \frac{1 }{j!} \sum_{l=0}^{j} {j \choose l} \lambda^{2l}
\mu^{2j-2l}
\frac{\prod_{i=1}^{1+j} (2i-1)\sqrt{\frac{\pi}{2}}}{2^{1+2j}\Gamma(1+j)} = \cr
= \sum_{a,b=0}^{+\infty} \frac{(-1)^a}{2^a a!}
 \frac{(-1)^b}{2^b b!}\sum_{j=0}^{+\infty}
 \frac{\prod_{i=1}^{1+j} (2i-1)\sqrt{\frac{\pi}{2}}}{j!2^{2j}\Gamma(1+j)}
 \sum_{l=0}^{j} {j \choose l} \lambda^{2l+2a}
\mu^{2j+2b-2l}.
}
\end{equation*}
Setting $n:=l+a$ and $%
m:=j+b-l$, we also have that
\begin{eqnarray}\nonumber
&& F(\lambda, \mu)  = \sum_{a,b=0}^{+\infty} \frac{(-1)^a}{2^a a!}
 \frac{(-1)^b}{2^b b!}\sum_{j=0}^{+\infty} \frac{\prod_{i=1}^{1+j} (2i-1)\sqrt{\frac{\pi}{2}}}{j!2^{2j}\Gamma(1+j)}
 \sum_{l=0}^{j} {j \choose l} \lambda^{2l+2a}
\mu^{2j+2b-2l} \\ \label{trallala}
&&=\sum_{n,m}^{} \sum_{j}^{}\frac{\prod_{i=1}^{1+j} (2i-1)\sqrt{\frac{\pi}{2}}}{j!2^{2j}\Gamma(1+j)}  \sum_{l=0}^{j}\frac{(-1)^{(n-l)}}{2^{n-l} {(n-l)}!}
 \frac{(-1)^{m+l-j}}{2^{m+l-j} {(m+l-j)}!}
  {j \choose l} \lambda^{2n}
\mu^{2m}.
\end{eqnarray}
Since $F(\lambda, \mu)=\E[W]$ from \paref{fico2}, on the other hand \paref{hermiteExp} yields
\begin{eqnarray}\nonumber
{ F(\lambda, \mu) }&=&\frac{1}{2\pi} \int_{\mathbb R^2} \sqrt{y^2 + z^2} \e^{\lambda y - \frac{\lambda^2}{2}} \e^{\mu z - \frac{\mu^2}{2}}\e^{-\frac{y^2+z^2}{2}}\,dy dz\\ \nonumber
&& = \frac{1}{2\pi} \int_{\mathbb R^2} \sqrt{y^2 + z^2} \sum_{a=0}^{+\infty} H_a(y) \frac{\lambda^a}{a!} \sum_{b=0}^{+\infty} H_b(z) \frac{\mu^b}{b!}\e^{-\frac{y^2+z^2}{2}}\,dy dz \\ \label{d def}
&& =  \sum_{a,b=0}^{+\infty} \left( \frac{1}{a!b!2\pi} \underbrace{\int_{\mathbb R^2} \sqrt{y^2 + z^2}  H_a(y) H_b(z) \e^{-\frac{y^2+z^2}{2}}\,dy dz}_{:= d(a,b)} \right)\lambda^a \mu^b.
\end{eqnarray}
By the same reasoning as above, if $a$ or $b$ is odd, then $d(a,b)$ in \paref{d def} must vanish. By combining the expansions in \paref{trallala} and \paref{d def}, we have
\begin{eqnarray}\nonumber
&&\alpha_{2n, 2m}
 = \frac{1}{2\pi} \int_{\mathbb R^2} \sqrt{y^2 + z^2}  H_{2n}(y) H_{2m}(z) \e^{-\frac{y^2+z^2}{2}}\,dy dz \\
&=& (2n)!(2m)!\frac{(-1)^{m+n}}{2^{n+m}}\sum_{j}^{}\frac{(-1)^{j}\prod_{i=1}^{1+j} (2i-1)
\sqrt{\frac{\pi}{2}}}{2^j j!\Gamma(1+j)}  \sum_{l=0}^{j}\frac{{j \choose l}}{ {(n-l)}! {(m+l-j)}!}. \label{sciuri}
\end{eqnarray}
The equality \eqref{e:alpha} now follows from \paref{sciuri} and some computations:
\begin{eqnarray*}
\alpha_{2n,2m}&=&\frac{1}{2\pi} \int_{\mathbb R^2} \sqrt{y^2 + z^2}  H_{2n}(y) H_{2m}(z) \e^{-\frac{y^2+z^2}{2}}\,dy dz \\
&=& (2n)!(2m)!\frac{(-1)^{m+n}}{2^{n+m}}\sum_{j}^{}(-1)^{j}
\frac{\prod_{i=1}^{1+j} (2i-1)\sqrt{\frac{\pi}{2}}}{2^j j!\Gamma(1+j)}
\sum_{l=0}^{j}\frac{{j \choose l}}{ {(n-l)}! {(m+l-j)}!} \\
&=& (2n)!(2m)!\frac{(-1)^{m+n}}{2^{n+m}}
\sum_{j}^{}(-1)^{j}\frac{(2j+1)!!\sqrt{\frac{\pi}{2}}}{2^j(j!)^2}
 \sum_{l=0}^{j}\frac{{j \choose l}}{ {(n-l)}! {(m+l-j)}!}\\
 &=&\frac{(2n)!(2m)!}{n! m!}\frac{(-1)^{m+n}}{2^{n+m}}
\sum_{j=0}^{n+m}(-1)^{j}\frac{(2j+1)!!\sqrt{\frac{\pi}{2}}}{2^j j!}
 \sum_{l=0}^{j}{ n\choose l}{ m\choose j-l} \\
&=&\sqrt{\frac{\pi}{2}}\frac{(2n)!(2m)!}{n! m!}\frac{(-1)^{m+n}}{2^{n+m}}
\sum_{j=0}^{n+m}(-1)^{j}\frac{(2j+1)!!}{2^j j!}
 { n+m\choose j} \\
&=&\sqrt{\frac{\pi}{2}}\frac{(2n)!(2m)!}{n! m!}\frac{(-1)^{m+n}}{2^{n+m}}
\sum_{j=0}^{n+m}(-1)^{j}\frac{(2(j+1))!}{2^{j+1} 2^j j! (j+1)!}
 { n+m\choose j} \\
&=&\sqrt{\frac{\pi}{2}}\frac{(2n)!(2m)!}{n! m!}\frac{(-1)^{m+n}}{2^{n+m}}
\sum_{j=0}^{n+m}(-1)^{j}\frac{(2j+1)!}{2^{2j} (j!)^2}
 { n+m\choose j}.
\end{eqnarray*} \end{proof}
{

%\begin{verbatim}
%``Point ... " should be something else? Numbering of Points should be dynamic. part?
%\end{verbatim}

We note that for two-dimensional random fields on the plane, the chaos decomposition of the length of level curves was derived earlier by Kratz and Le{\'o}n, see \cite{KL}; our derivation of the projection coefficients in Lemma \ref{Euclidean norm} is different from theirs (albeit equivalent, by uniqueness), and it was hence reported for the sake of completeness.

\noindent\begin{proof}[Proof of Proposition \ref{teoexp}]
In view of Definition \ref{d:chaos}, the computations in Lemma
\ref{indicator function} and Lemma \ref{Euclidean norm} (together with the fact that the three
random variables $T_n(x),\, \partial_1 \widetilde T_n(x)$ and $\partial_2 \widetilde T_n(x)$
are stochastically independent, as recalled above) show that, for fixed $x\in \mathbb{T}$,
the projection of the random variable
$$
\frac{1}{2\varepsilon} 1_{[-\varepsilon, \varepsilon]}(T_n(x)) \sqrt{\partial_1 \widetilde
T_n(x)^2+\partial_2 \widetilde T_n(x)^2}
$$
onto each odd chaos vanishes, whereas the projection onto the chaos $C_{2q}$, for $q\geq 1$, equals
$$
\sum_{u=0}^{q}\sum_{m=0}^{u}
\frac{\alpha _{2m,2u-2m}\beta^{\varepsilon} _{2q-2u}
}{(2m)!(2u-2m)!(2q-2u)!}  H_{2q-2u}(T_n(x))
H_{2m}(\partial_1 \widetilde T_n(x))H_{2u-2m}(\partial_2
\widetilde T_n(x)).
$$ Since $\int_{\mathbb{T}}dx<\infty$, standard arguments based on Jensen's
inequality and dominated convergence yield that
$\Lc^{\varepsilon}_n[q]=0$ if $q$ is odd and
for every $q\geq 1$,
\begin{eqnarray*}
&&\Lc^{\varepsilon}_n[2q]\\
&&= \sqrt{\frac{4\pi^2n}{2}}\sum_{u=0}^{q}\sum_{m=0}^{u}
\frac{\alpha _{2m,2u-2m}\beta^{\varepsilon} _{2q-2u}
}{(2m)!(2u-2m)!(2q-2u)!}\times \\
&& \hspace{4cm} \times\!\!\int_{\mathbb T}\!\! H_{2q-2u}(T_n(x))
H_{2m}(\partial_1 \widetilde T_n(x))H_{2u-2m}(\partial_2
\widetilde T_n(x))\,dx.\notag
\end{eqnarray*}
In view of Lemma \ref{approx} and \paref{eq:len chaos decomp} one has that for every $q\ge 0$,
as $\varepsilon \to 0$, $\Lc^\varepsilon_n[q]$
necessarily converge to $\Lc_n[q]$ in $L^2$. We just proved that
$$\Lc^\varepsilon_n[q]=0$$ for $q=2m+1$ as stated in part (a) of Proposition
\ref{teoexp}.
Moreover, using \eqref{e:sat}, we deduce from this fact that representation \eqref{e:pp} in part (b) of Proposition
\ref{teoexp} is valid. To complete the proof of part (a) of Proposition \ref{teoexp}, we need first to show that that
$\Lc_n[2]=0.$
From the previous discussion we deduce that $\Lc_n[2]$ equals
\begin{equation}
\label{eq:L proj C2}
\begin{split}
\Lc_n[2]&=
\sqrt{4\pi^2}\sqrt{\frac{n}{2}}\bigg(
\frac{\alpha_{0,0}\beta_{2}}{2} \int_{\mathbb T}
H_{2}(T_n(x))\,dx + \frac{\alpha_{0,2}\beta_{0}}{2}
\int_{\mathbb T}  H_{2}(\partial_2 \widetilde
T_n(x))\,dx \\&+  \frac{\alpha_{2,0}\beta_{0}}{2}
\int_{\mathbb T} H_{2}(\partial_1 \widetilde T_n(x))\,dx
\bigg).
\end{split}
\end{equation}
Since $H_2(x)=x^2-1$, we may write
\begin{equation}
\label{eq:int H2der expl}
\begin{split}
\int_{\mathbb T} H_{2}(T_n(x))\,dx &= \int_{\mathbb T}
\left ( T_n(x)^2 -1                           \right)\,dx=
\int_{\mathbb T}
\left ( \frac{1}{\mathcal N_n}
\sum_{\lambda, \lambda'\in \Lambda_n} a_\lambda \overline a_{\lambda'}
e_{\lambda - \lambda'}(x)
 -1                           \right)\,dx\\
&=
 \frac{1}{\mathcal N_n}\sum_{\lambda, \lambda'\in \Lambda_n} a_\lambda \overline a_{\lambda'}
 \underbrace{\int_{\mathbb T}e_{\lambda - \lambda'}(x)\,dx}_{\delta_{\lambda}^{\lambda'}}
-1= \frac{1}{\mathcal N_n} \sum_{\lambda_\in \Lambda_n} (|a_\lambda|^2 -1),
\end{split}
\end{equation}
where $\delta_{\lambda}^{\lambda'}$ is the Kronecker symbol. (Observe that $\E[|a_\lambda|^2]=1$, hence the expected value of the integral
$\int_{\mathbb T} H_{2}(T_n(x))\,dx$ is $0$, as expected.) Analogously, for $j=1,2$ we have
\begin{equation}
\label{eq:int H2 expl}
\begin{split}
\int_{\mathbb T} H_{2}(\partial_j \widetilde T_n(x))\,dx
&= \int_{\mathbb T} \left ( \frac{2}{n}\frac{1}{\mathcal N_n}
\sum_{\lambda, \lambda'\in \Lambda_n} \lambda_j \lambda'_j
a_\lambda \overline a_{\lambda'} e_{\lambda - \lambda'}(x )
 -1                           \right)\,dx\\
 &=\frac{2}{n}\frac{1}{\mathcal N_n}
\sum_{\lambda\in \Lambda_n} \lambda_j^2 |a_\lambda|^2 - 1
=\frac{1}{\mathcal N_n} \frac{2}{n} \sum_{\lambda \in \Lambda_n}
\lambda_j^2 (|a_\lambda|^2 - 1),
\end{split}
\end{equation}
where the used Lemma \ref{lemmavar} to establish the last equality.

Since $\alpha_{2n,2m}=\alpha_{2m,2n}$, and in light of \eqref{eq:int H2der expl} and \eqref{eq:int H2 expl},
we may rewrite \eqref{eq:L proj C2} as
\begin{equation}
\label{eq:L proj C2 expl}
\begin{split}
&\Lc_n[2]=\sqrt{4\pi^2}\sqrt{\frac{n}{2}}\left( \frac{\alpha_{0,0}\beta_{2}}{2}
\frac{1}{\mathcal N_n} \sum_{\lambda_\in \Lambda_n} (|a_\lambda|^2 -1)
+ \frac{\alpha_{0,2}\beta_{0}}{2}
\frac{1}{\mathcal N_n} \frac{2}{n} \sum_{\lambda \in \Lambda_n}
\underbrace{(\lambda_1^2+\lambda_2^2)}_{=n} (|a_\lambda|^2 - 1)
  \right)\\
&=\sqrt{4\pi^2}\sqrt{\frac{n}{2}}\frac{1}{2\mathcal N_n}
\left( \alpha_{0,0}\beta_{2}
 \sum_{\lambda_\in \Lambda_n} (|a_\lambda|^2 -1)
+ 2\alpha_{0,2}\beta_{0}
 \sum_{\lambda \in \Lambda_n}
 (|a_\lambda|^2 - 1)
  \right)\\
&=\sqrt{4\pi^2}\sqrt{\frac{n}{2}}\frac{1}{2\mathcal N_n}
\left( \alpha_{0,0}\beta_{2}+ 2\alpha_{0,2}\beta_{0}\right)
 \sum_{\lambda_\in \Lambda_n} (|a_\lambda|^2 -1).
\end{split}
\end{equation}
Since
$$
\alpha_{0,0} = \sqrt{\frac{\pi}{2}},\quad \alpha_{0,2}=\alpha_{2,0} = \frac12
\sqrt{\frac{\pi}{2}}, \quad \beta_0 =\frac{1}{ \sqrt{2\pi}}, \quad \beta_2 = - \frac{1}{ \sqrt{2\pi}},
$$
we have that
$$
 \alpha_{0,0}\beta_{2}+ 2\alpha_{0,2}\beta_{0}=0
$$
and hence $\Lc_n[2]=0$ from \paref{eq:L proj C2 expl}.
The proof of Proposition \ref{teoexp} is hence concluded, in view of \eqref{eq:len chaos decomp}.
\end{proof}

\section{Proofs of Proposition \ref{prop:main result on proj4} and Proposition \ref{eq:Lproj4 replace L}
}
\label{dimProp}

%\begin{proof}[Proof of Proposition \ref{eq:Lproj4 replace L} assuming Proposition \ref{teoexp} and Proposition \ref{prop:main result on proj4}]
%
%When studying the asymptotic variance of $\text{proj}(\Lc_n|C_4)$, it should be noted that, because we are inside a fixed chaotic component, convergence in distribution is equivalent to convergence of the moments of all order. Hence to derive the asymptotic variance of the fourth order chaotic component it is enough to compute the variance of the limiting distribution; this computation is given in Proposition \ref{prop:main result on proj4}, Eq.\ref{ultimo}.
%
%
%\begin{verbatim}
%Here we need to check that the asymptotic variance of the projection into
%$4$th Wiener Chaos is the same as the total one.
%\end{verbatim}
%
%\end{proof}

One of the main findings of the present paper is that, for any sequence $\lbrace n_j\rbrace$ such that
$\Nc_{n_j}\to \infty$ and $|\widehat \mu_{n_j}(4)|$ converges, the distribution of the normalised sequence
$\{ \widetilde{\Lc}_{n_j}\}$ in \eqref{e:culp} is asymptotic to one of its
fourth-order chaotic projections. The aim of this section is a precise analysis of the asymptotic behavior of the sequence
\begin{equation*}\label{e:ron}
 \frac{ \mathcal{L}_{n_j}[4]}
{\sqrt{\Var(\mathcal{L}_{n_j}[4] )}} , \quad j\geq 1.
\end{equation*}
which will allow us to prove Proposition \ref{prop:main result on proj4} and Proposition \ref{eq:Lproj4 replace L}.

\subsection{Preliminary results}

Here we state the key tools for our proofs: first an explicit formula for $\mathcal{L}_{n_j}[4]$ and then a Central Limit Theorem for some of its ingredients.
%Let us introduce some more notation. We set
%\begin{equation}\label{e:pzi}
%\psi(\eta ) := \frac{3+\eta}{8}, \quad \eta\in [-1,1].
%\end{equation}
First we need some intermediate results, whose proofs follow immediately from the fact that, for every $n\in S$,
$$
\widehat{\mu}_n(4) = \frac{1}{n^2 \Nc_n} \sum_{\lambda=(\lambda_1, \lambda_2)\in \Lambda_n} \left(\lambda_1^4 +\lambda_2^4-6\lambda_1^2\lambda_2^2\right),
$$
as well as from the identity $\lambda_1^2 + \lambda_2^2 = n$ for $\lambda=(\lambda_1, \lambda_2)\in \Lambda_n$, and  the following elementary symmetry consideration:
$$
 \frac{1}{n^2 \Nc_n} \sum_{\lambda=(\lambda_1, \lambda_2)\in \Lambda_n} \lambda_1^4 =  \frac{1}{n^2 \Nc_n} \sum_{\lambda=(\lambda_1, \lambda_2)\in \Lambda_n} \lambda_2^4.
$$
\begin{lemma}\label{semplice}
For every $n\in S$
we have
\begin{equation*}
\frac{1}{n^2 \mathcal N_n} \sum_{\lambda=(\lambda_1,\lambda_2)\in  \Lambda_n} \lambda_\ell^4 = \frac{3+\widehat \mu_n(4)}{8},
\end{equation*}
where $\ell=1,2$, and moreover
\begin{equation*}
\frac{1}{n^2 \mathcal N_n} \sum_{\lambda=(\lambda_1,\lambda_2)\in  \Lambda_n} \lambda_1^2 \lambda_2^2 = \frac{1-\widehat \mu_n(4)}{8}.
\end{equation*}
\end{lemma}
Let us state now the above mentioned CLT result.
Let us define, for $n\in S$,
\begin{eqnarray}
W(n):=\left(
\begin{array}{c}
W_{1}(n) \\
W_{2}(n) \\
W_{3}(n) \\
W_{4}(n)
\end{array}%
\right ) & :=& \frac{1}{n\sqrt{\Nc_{n}/2}}\sum_{\substack{\lambda = (\lambda_1,\lambda_2)\in \Lambda_{n} \\ \lambda_2> 0 }}\left(
|a_{\lambda }|^{2}-1\right) \left(
\begin{array}{c}\label{Wdef}
n \\
\lambda _{1}^{2} \\
\lambda _{2}^{2} \\
\lambda _{1}\lambda _{2}%
\end{array}%
\right).
\end{eqnarray}

%\begin{verbatim}
%I would combine the following couple of lemmas into a single one and prove it in the next section.
%In fact, it could be part of the Wiener Chaos Proposition 3.2
%\end{verbatim}

Exploiting the representation \eqref{e:pp} in the case $q=2$, one can show the following.
\begin{lemma}\label{p:formula4}
We have, for diverging subsequences $\lbrace n_j \rbrace \subseteq S$ such that $\mathcal N_{n_j}\to +\infty$ and $\widehat \mu_{n_j}(4)$ converges,
\begin{equation}
\mathcal{L}_{n_j}[4]=\sqrt{\frac{E_{n_j}}{ 512\, \mathcal N^2_{n_j}}}\Big(1 + W_1(n_j) ^2-2W_2(n_j)^2-2W_3(n_j)^2 - 4W_4(n_j)^2+ o_\P (1) \Big).
\end{equation}
%where, the sequence of random variables $\{R_{n_j}\}$ converges in probability to $1$.
\end{lemma}

The proof of Lemma \ref{p:formula4} will be given in \S \ref{sec:proof expl proj C2}.

\begin{lemma}\label{p:clt} Assume that the subsequence $\{n_j\}\subseteq S$ is such that $\Nc_{n_j} \to +\infty$ and $\widehat{\mu}_{n_j}(4) \to \eta\in [-1,1]$.
%and denote $\psi=\psi(\eta)$ as in \eqref{e:pzi}.
Then, as $n_j\to \infty$, the following CLT holds:
\begin{eqnarray}\label{e:punz}
W(n_j) \stackrel{\rm d}{\longrightarrow} Z(\eta)= \left(
\begin{array}{c}
Z_{1} \\
Z_{2} \\
Z_{3} \\
Z_{4}%
\end{array}%
\right ),
\end{eqnarray}
where $Z(\eta)$ is a centered Gaussian vector with covariance
\begin{equation}\label{e:sig}
\Sigma=\Sigma(\eta) =\left(
\begin{array}{cccc}
1 & \frac{1}{2} & \frac{1}{2} & 0 \\
\frac{1}{2} & \frac{3+\eta}{8}  & \frac{1-\eta}{8}  & 0 \\
\frac{1}{2} & \frac{1-\eta}{8}  & \frac{3+\eta}{8}  & 0 \\
0 & 0 & 0 & \frac{1-\eta}{8}
\end{array}%
\right).
\end{equation}%
The eigenvalues of $\Sigma $ are $ 0,\frac{3}{2},%
\frac{1-\eta}{8},\frac{1+\eta}{4}$ and hence, in particular, $\Sigma$ is singular.
\end{lemma}

\noindent\begin{proof}{According to \cite[Theorem 6.2.3]{N-P}, in order to achieve the desired conclusion it is sufficient to prove the following relations: (a) for every fixed integer $n_j$, each component of the vector $W(n_j)$ is an element of the second Wiener chaos associated with ${\bf A}$ (see \S \ref{ss:berryintro}), (b) as $n_j\to \infty$, the covariance matrix of $W(n_j)$ converges to $\Sigma$, and (c) for every $k=1,2,3,4$, as $n_j\to \infty$, one has that $W_k(n_j)$ converges in distribution to a one-dimensional centered Gaussian random variable. Part (a) is trivially verified. Part (b) follows by a direct computation based on Lemma \ref{semplice},
as well as on the fact that the random variables in the set $$\left\{
|a_{\lambda }|^{2}-1  : \lambda \in \Lambda_{n_j}, \, \lambda_2> 0 \right\}$$ are centered,
independent, identically distributed and with unit variance. To prove part (c),
write $\Lambda_{n_j}^+ := \{ \lambda \in \Lambda_{n_j}, \, \lambda_2> 0\}$ and observe that,
for every $k$ and every $n_j$, the random variable $W_k(n_j)$ is of the form
$$
W_k(n_j) = \sum_{\lambda \in \Lambda_{n_j}^+} c_k(n_j, \lambda)\times  (|a_{\lambda }|^{2}-1)
$$
where $\{c_k(n_j, \lambda)\}$ is a collection of positive deterministic coefficients
such that $$\max_{\lambda \in  \Lambda_{n_j}^+} c_k(n_j, \lambda)\to 0,$$ as $n_j\to\infty$.
An application of the Lindeberg criterion, e.g. in the quantitative form stated
in \cite[Proposition 11.1.3]{N-P}, yields that $W_k(n_j)$ converges in distribution
to a Gaussian random variable, and therefore that (c) also holds. This proves \eqref{e:punz}. Since it is easy to verify
the claimed eigenvalues of $\Sigma$ in \paref{e:sig} via an explicit computation, this concludes the proof of Lemma \ref{p:formula4}.
}
\end{proof}

\subsection{Proof of Proposition \ref{prop:main result on proj4}:
asymptotic behaviour of $\Lc_n[4]$}\label{proof4}

\begin{proof}[Proof of Proposition \ref{prop:main result on proj4} assuming Lemma \ref{p:formula4}]\label{p:4nclt} Let $\{n_j\} \subseteq S$ be such that $\Nc_{n_j}\to \infty$ and $|\widehat{\mu}_{n_j}(4)|\to \eta\in [0,1]$. For each subsequence $\{ n'_j \}\subseteq \{n_j\}$ there exists a subsubsequence $\{n''_j \}\subseteq \{n'_j\}$ such that it holds either (i) $\widehat{\mu}_{n''_j}(4)\to \eta$ or
(ii) $\widehat{\mu}_{n''_j}(4)\to -\eta$. Set
\begin{equation*}
v(n''_j) :=\sqrt{\frac{E_{n''_j}}{512\,\Nc^2_{n''_j}}}, \quad j\geq 1.
\end{equation*}
Then, as $n''_j\to \infty$
\begin{equation}\label{e:4nclt}
Q(n''_j) := \frac{\mathcal{L}_{n''_j}[4]}{v(n''_j) } \stackrel{\rm d}{\longrightarrow} 1+Z_1^2-2Z_2^2-2Z_3^2-4Z_4^2,
\end{equation}
by Lemma \ref{p:clt} and Lemma \ref{p:formula4};
here $Z = Z(\eta)\in \R^{4}$ is as in \eqref{e:punz}, i.e. a centred Gaussian $4$-variate vector with covariance matrix as in \paref{e:sig}.
%and if (i) holds, then $\psi(\eta) = \frac{3+\eta}{8}$ otherwise $\psi(\eta) = \frac{3-\eta}{8}$.

%\begin{verbatim}
%I don't understand what is tilde{Z} and why we need it.
%Is it only for the variance computation?
%Is it a simple linear transformation of Z (doesn't say but indication
%in the variance computation)?
%\end{verbatim}

\noindent Actually, the multidimensional CLT stated in \eqref{e:punz} implies that
$$
(W_1(n_j)^2, W_2(n_j)^2, W_3(n_j)^2, W_4(n_j)^2) \stackrel{\rm d}{\longrightarrow} (Z_1^2, Z_2^2, Z_3^2, Z_4^2).
$$
%Let us now introduce the centered Gaussian vector $\widetilde{Z}^{\top} := (\widetilde{Z}_{1}, \widetilde{Z}_{2}, \widetilde{Z}_{3}, \widetilde{Z}_{4})$, with %covariance matrix given by
%$$
%\widetilde{\Sigma }%
%:=\left(
%\begin{array}{cccc}
%1 & \frac{1}{2\sqrt{\psi }} & \frac{1}{2\sqrt{\psi }} & 0 \\
%\frac{1}{2\sqrt{\psi }} & 1 & \frac{1}{2\psi }-1 & 0 \\
%\frac{1}{2\sqrt{\psi }} & \frac{1}{2\psi }-1 & 1 & 0 \\
%0 & 0 & 0 & 1%
%\end{array}%
%\right) \text{ .}
%$$
A simple computation of Gaussian moments now yields%
\begin{eqnarray}\nonumber
&& \Var\left ( {Z}_{1}^{2}-2{Z}_{2}^{2}-2{Z}_{3}^{2}-4{Z}_{4}^{2}\right )
\\ \nonumber
%&& =\Var\left ( \widetilde{Z}_{1}^{2}-2\psi \widetilde{Z}_{2}^{2}-2\psi
%\widetilde{Z}_{3}^{2}-4\left (\frac{1}{2}-\psi\right )\widetilde{Z}_{4}^{2}\right )  \\ \nonumber
%&&=\Var\left ( H_{2}(\widetilde{Z}_{1})-2\psi H_{2}(\widetilde{Z}_{2})-2\psi
%H_{2}(\widetilde{Z}_{3})-4\left (\frac{1}{2}-\psi \right )H_{2}(\widetilde{Z}%
%_{4})\right ) \\ \nonumber
%&&=2+8\psi ^{2}+8\psi ^{2}+32\left (\frac{1}{2}-\psi \right )^{2}-4\psi \Cov\left (H_{2}(%
%\widetilde{Z}_{1}),H_{2}(\widetilde{Z}_{2})\right )  \\ \nonumber
%&&\quad-4\psi \Cov\left ( H_{2}(\widetilde{Z}_{1}),H_{2}(\widetilde{Z}_{3})\right )
%+8\psi ^{2} \Cov\left ( H_{2}(\widetilde{Z}_{2}),H_{2}(\widetilde{Z}_{3})\right )
\\ \nonumber
&&=2+8\left (\frac{3+\eta}{8} \right )^{2}+8\left (\frac{3+\eta}{8} \right )^{2}+32\left (\frac{1-\eta}{8} \right  )^{2}-2-2+4\left (\frac{1-\eta}{4} \right )^{2}
%\\ \nonumber
%&&=2+8\psi ^{2}+8\psi ^{2}+32\left (\frac{1}{4}+\psi ^{2}-\psi \right )-2-2+16\psi ^{2}\left (%
%\frac{1}{4\psi ^{2}}+1-\frac{1}{\psi } \right )
%\\
%&&=64\psi ^{2}-48\psi +10.\label{ultimo}
= \eta^2+1,\label{e:magic}
\end{eqnarray}
%Now a simple computation yields
%\begin{equation}
%64\, \psi(\eta)^2-48\, \psi(\eta)+10
%\end{equation}
entailing in particular that $\Var\left ( {Z}_{1}^{2}-2{Z}_{2}^{2}-2{Z}_{3}^{2}-4{Z}_{4}^{2}\right )$ is the same in both cases (i) -- (ii). We can rewrite \paref{e:4nclt} as, for $n''_j\to +\infty$,
\begin{equation*}
\frac{ \mathcal{L}_{n''_j}[4]}{\sqrt{1+\eta^2}\,v(n''_j) } \stackrel{\rm d}{\longrightarrow} \frac{1}{\sqrt{1+\eta^2}}(1+Z_1^2-2Z_2^2-2Z_3^2-4Z_4^2).
\end{equation*}
We claim that, in both cases (i)--(ii), the random variable
$$
\frac{1}{\sqrt{1+\eta^2}}\Big( 1+Z_1^2-2Z_2^2-2Z_3^2-4Z_4^2 \Big)
$$
has the same law as $\mathcal{M}_{\eta}$, as defined in $\eqref{e:r}$.

To verify this, let $\widetilde \Sigma=\widetilde \Sigma(\eta)$ be the covariance matrix of $(\widetilde Z_1,\widetilde Z_2,\widetilde Z_3)$, where $\widetilde Z_i = \frac{Z_i}{\sqrt{\Var(Z_i)}}$ for $i=1,2,3$, i.e.
\begin{equation*}
\widetilde \Sigma(\eta)=\left(\begin{matrix}
1 & \frac{\sqrt{2}}{\sqrt{3+\eta }} & \frac{\sqrt{2}}{\sqrt{3+\eta }} \\
\frac{\sqrt{2}}{\sqrt{3+\eta }} & 1 & \frac{1-\eta}{3+\eta} \\
\frac{\sqrt{2}}{\sqrt{3+\eta }} & \frac{1-\eta}{3+\eta} & 1
\end{matrix}
\right).
\end{equation*}
We diagonalize $\widetilde \Sigma$:
\begin{equation*}
\widetilde \Sigma = ADA^{t},
\end{equation*}
where $A$ is the orthogonal matrix
\begin{equation*}
A = \left(\begin{matrix}
\frac{\sqrt{3+\eta}}{\sqrt{\eta+7}} &0 &-\frac{2}{\sqrt{\eta+7}} \\
\frac{\sqrt{2}}{\sqrt{\eta+7}}  &\frac{1}{\sqrt{2}} &\frac{\sqrt{3+\eta}}{\sqrt{2}\sqrt{\eta+7}}\\
\frac{\sqrt{2}}{\sqrt{\eta+7}} &-\frac{1}{\sqrt{2}} &\frac{\sqrt{3+\eta}}{\sqrt{2}\sqrt{\eta+7}}
\end{matrix} \right),
\end{equation*}
and
\begin{equation*}
D = \operatorname{diag}\left(\frac{\eta+7}{3+\eta}, \frac{2(1+\eta)}{3+\eta}, 0\right).
\end{equation*}
Hence
$\widetilde Z_{1}=X_{1}$, $$\widetilde Z_{2}=\frac{\sqrt{2}}{\sqrt{3+\eta}}X_{1} + \frac{\sqrt{1+\eta}}{\sqrt{3+\eta}}X_2,$$
$$\widetilde Z_{3}=\frac{\sqrt{2}}{\sqrt{3+\eta}}X_{1} - \frac{\sqrt{1+\eta}}{\sqrt{3+\eta}}X_{2},$$
where $X=(X_{1},X_{2})\in\R^{2}$ is a bivariate standard Gaussian random vector.
Let us now set
\begin{equation*}
\psi=\frac{3+\eta}{8}.
\end{equation*}
Adding $X_{3}$ as one more standard Gaussian, independent of $(X_1,X_2)$,  we find
\begin{equation*}
\begin{split}
&\frac{1}{\sqrt{1+\eta^2}}  (1+Z_1^2-2Z_2^2-2Z_3^2-4Z_4^2)\cr
&= \frac{1}{\sqrt{64\psi^2 -48\psi + 10}} \bigg (
X_{1}^{2}-2\psi \left( \frac{\sqrt{2}}{\sqrt{3+\eta}}X_{1}+\frac{\sqrt{1+\eta}}{\sqrt{3+\eta}}X_{2}\right)^{2}-
2\psi \left( \frac{\sqrt{2}}{\sqrt{3+\eta}}X_{1}-\frac{\sqrt{1+\eta}}{\sqrt{3+\eta}}X_{2}\right)^{2}
\\&-4\left (\frac{1}{2}-\psi \right )X_{3}^{2} + 1 \bigg )
\\&= \frac{1}{\sqrt{1+\eta^{2}}}
\left (X_{1}^{2}-\frac{1}{2}(3+\eta) \left( \frac{2}{3+\eta}X_{1}^{2}+\frac{1+\eta}{3+\eta}X_{2}^{2}\right)
-\frac{1}{2}(1-\eta)X_{3}^{2} + 1 \right )
\\&= \frac{1}{2\sqrt{1+\eta^{2}}}\left(2-(1+\eta)X_{2}^{2} -(1-\eta)X_{3}^{2} \right ) \mathop{=}^d \mathcal M_\eta.
\end{split}
\end{equation*}
In particular, we have therefore that
\begin{equation}\label{newnice}
K(n_j) := \frac{ \mathcal{L}_{n_j}[4]}{\sqrt{1+\eta^2}\,v(n_j) } \stackrel{\rm d}{\longrightarrow} \mathcal M_\eta.
\end{equation}
To conclude the proof, we observe that, since $\{K(n_j)\}$ is a sequence of random variables belonging to a fixed Wiener chaos and converging in distribution, one has necessarily that (by virtue e.g. of \cite[Lemma 2.1-(ii)]{N-R})
$$
\sup_{n_j}\mathbb{E} | K(n_j)|^p <\infty, \quad \forall p>0.
$$
Standard arguments based on uniform integrability yield therefore that, as $n_j\to\infty$,
$$
\frac{{\rm Var}(\mathcal{L}_{n_j}[4])}{v(n_j)^2(1+\eta^2)} = \E(K(n_j)^2) \to  \E(\mathcal{M}_\eta^2) =1,
$$
which is the same as \eqref{perlavar}.

\end{proof}

\subsection{Proof of Proposition \ref{eq:Lproj4 replace L}: $\mathcal{L}_{n_j}[4]$ dominates $\mathcal L_{n_j}$}

Now we are able to prove one of the main findings in this paper, i.e. that the fourth-chaotic projection and the total nodal length have the same asymptotic behavior.

\smallskip

\noindent \begin{proof}[Proof] %of Proposition \ref{eq:Lproj4 replace L}]
Let us first prove \paref{vaar}.
Note that \paref{eq:var leading KKW2} and
 Proposition \ref{prop:main result on proj4} immediately give \paref{vaar} i.e., as $\mathcal N_{n_j}\to +\infty$,
\begin{equation*}
\Var(\mathcal L_{n_j})\sim \Var(\mathcal L_{n_j}[4]).
\end{equation*}
%Now let us show that \paref{vaar} is equivalent to \paref{vaar2}.
%
Now, since different chaotic projections are orthogonal in $L^2$, from part (b) of Proposition \ref{teoexp} we have
\begin{equation}\label{varexp}
\Var(\mathcal L_{n_j})  = \Var(\mathcal L_{n_j}[4]) + \sum_{q=3}^{+\infty} \Var(\mathcal L_{n_j}(2q)).
\end{equation}
%%
%or, equivalently,
%\begin{equation}\label{ciao}
%\var(\mathcal L_{n_j} - \mathcal L_{n_j}[4]) = \sum_{q=3}^{+\infty} \Var(\mathcal L_{n_j}(2q)).
%\end{equation}
%%
% \paref{vaar} and \paref{varexp} entail that
%\begin{equation}\label{intermediate}
%\sum_{q=3}^{+\infty} \Var(\mathcal L_{n_j}(2q)) = o\left (\var(\mathcal L_{n_j}[4])    \right ), %=o\left ( \frac{E_{n_j}}{\mathcal N_{n_j}^2}  \right ),
%\end{equation}
%hence using \paref{perlavar} in \paref{intermediate} we get \paref{vaar2}.
%
%
%Conversely, if \paref{vaar2} holds, keeping in mind \paref{perlavar} from \paref{ciao} we have \paref{intermediate}. Let us substitute \paref{intermediate} in \paref{varexp}: we immediately have \paref{vaar}.
%
Dividing both sides of \paref{varexp} by $\Var(\mathcal L_{n_j}[4])$, we immediately  conclude the proof of Proposition \ref{eq:Lproj4 replace L}.

\end{proof}

\section{Proof of Lemma \ref{p:formula4}: explicit formula for $\mathcal{L}_{n_j}[4]$}\label{lemmaFor}

\label{sec:proof expl proj C2}

Consider the following representation of $\mathcal{L}_{n_j}[4]$, that is a particular case $q=2$ of \eqref{e:pp}:
\begin{eqnarray}\label{e:nus}
\mathcal{L}_{n_j}[4]&=&\sqrt{4\pi^2}\sqrt{\frac{n}{2}}\Big(\frac{\alpha_{0,0}\beta_4}{4!}\int_{\mathbb
T} H_4(T_n(x))\,dx\\ \notag
&&+\frac{\alpha_{0,4}\beta_0}{4!}\int_{\mathbb T} H_4(\partial_2
\widetilde T_n(x))\,dx
+\frac{\alpha_{4,0}\beta_0}{4!}\int_{\mathbb T} H_4(\partial_1
\widetilde T_n(x))\,dx+\\ \notag
&&+\frac{\alpha_{0,2}\beta_2}{2!2!}\int_{\mathbb T}
H_2(T_n(x))H_2(\partial_2 \widetilde
T_n(x))\,dx\\
\notag
&&+\frac{\alpha_{2,0}\beta_2}{2!2!}\int_{\mathbb T}
H_2(T_n(x))H_2(\partial_1 \widetilde
T_n(x))\,dx+\\ \notag
&& +\frac{\alpha_{2,2}\beta_0}{2!2!}\int_{\mathbb T}
H_2(\partial_1 \widetilde T_n(x))H_2(\partial_2 \widetilde
T_n(x))\,dx \Big),
\end{eqnarray}
where the coefficients $\alpha_{\cdot, \cdot}$ and $\beta_{\cdot}$ are defined according to equation \eqref{e:alpha} and equation \eqref{e:beta}, respectively.

\subsection{Auxiliary results}

The next four lemmas yield a useful representation for the six summands appearing on the right-hand side of \eqref{e:nus}. In what follows, $n\in S$ %always stands for a positive integer
and, moreover, to simplify the discussion we will sometimes use the shorthand
\begin{eqnarray*}
&& \sum_{\lambda} = \sum_{\substack{\lambda = (\lambda_1,\lambda_2)\in \Lambda_{n} }}, \quad\quad \sum_{\lambda, \lambda'} = \sum_{\substack{\lambda, \lambda' \in \Lambda_{n} }}  \quad \quad\mbox{and}\quad \quad \sum_{\lambda : \lambda_2> 0} = \sum_{\substack{\lambda = (\lambda_1,\lambda_2)\in \Lambda_{n} \\ \lambda_2> 0 }},
\end{eqnarray*}
in such a way that the exact value of the integer $n$ will always be clear from the context. Also, the symbol $\{n_j\}$ will always denote a subsequence of integers contained in $S$ such that $\Nc_{n_j}\to \infty$ and $\widehat{\mu}_{n_j} (4) \to \eta\in [-1,1]$, as $n_j\to \infty$. As before, we write `$\stackrel{\P}\longrightarrow$' to denote convergence in probability, and we use the symbol $o_{\P}(1)$ to denote a sequence of random variables converging to zero in probability, as $\mathcal{N}_n\to\infty$.

Following \cite{KKW}, we will abundantly use the fine structure of
the {\it length-$4$ spectral correlation set}:
\begin{equation}
\label{eq:S4 def}
S_n(4) := \lbrace  (\lambda, \lambda', \lambda'', \lambda''')\in (\Lambda_n)^4 :
\lambda + \dots + \lambda''' = 0          \rbrace.
\end{equation}

\begin{lemma}[\cite{KKW}, p. $31$]
\label{lem:S4 fine struct}
Let $S_{n}(4)$ be the length-$4$ spectral correlation set defined in \eqref{eq:S4 def}. Then
$S_{n}(4)$ is the disjoint union
\begin{equation*}
S_{n}(4) = A_{n}(4)\cup B_{n}(4),
\end{equation*}
where $A_{n}(4)$ is all the $3$ permutations of
\begin{equation*}
\widetilde{A}_{n}(4) = \{(\lambda,\lambda',-\lambda,-\lambda'):\: \lambda,\lambda'\in \Lambda_{n},\, \lambda\ne\lambda'\},
\end{equation*}
and $B_{n}(4)$ is all the $3$ permutations of
\begin{equation*}
\widetilde{B}_{n}(4) = \{(\lambda,\lambda,-\lambda,-\lambda):\: \lambda \in \Lambda_{n}\}.
\end{equation*}
In particular, using the inclusion-exclusion principle,
\begin{equation*}
| S_{n}(4)| =3\Nc_{n}(\Nc_{n}-1).
\end{equation*}
\end{lemma}

\begin{lemma}\label{lem1} One has the following representation:
\begin{equation}\label{aaa}
\int_{\mathbb T} H_{4}(T_n(x))\,dx
=\frac{6}{\Nc_{n}}\left (\frac{1}{\sqrt{\Nc_{n}/2}}\sum_{\lambda :\lambda
_{2}> 0}(|a_{\lambda }|^{2}-1) + o_{\P}(1)  \right )^{2}-\frac{3}{\Nc_{n}^{2}}\sum_{\lambda
}|a_{\lambda }|^{4}.
\end{equation}
Also, as $ {n_j}\to\infty$,
\begin{equation}\label{bbb}
\frac{3}{\Nc_{n_j}}\sum_{\lambda
}|a_{\lambda }|^{4}\stackrel{\P}\longrightarrow 6.
\end{equation}

\end{lemma}
\noindent\begin{proof}
Using the explicit expression $H_4(x)=x^4 - 6x^2 +3$, we deduce that
\begin{align}
%\begin{split}
\nonumber
\int_{\mathbb T} H_{4}(T_n(x))\,dx &= \int_{\mathbb T}
\left ( T_n(x)^4 -6T_n(x)^2  +3 \right)\,dx\\
\nonumber
&= \frac{1}{\mathcal N_n^2}
\sum_{\lambda, \dots, \lambda''' \in \Lambda_n} a_{\lambda} \overline a_{\lambda'} a_{\lambda''}
\overline a_{\lambda'''}
\int_{\mathbb T}
\exp(2\pi i\langle \lambda -\lambda' + \lambda'' - \lambda''', x\rangle )\,dx+\\
\nonumber
 & -\, 6\frac{1}{\mathcal N_n}
\sum_{\lambda, \lambda' \in \Lambda_n} a_{\lambda}
\overline a_{\lambda'}
\int_{\mathbb T}\exp(2\pi i\langle \lambda - \lambda', x\rangle )\,dx   +3 \\
\label{eq:int H4 expl}
&= \frac{1}{\mathcal N_n^2}
\sum_{ \lambda - \lambda' + \lambda'' - \lambda'''=0} a_{\lambda} \overline a_{\lambda'} a_{\lambda''}
\overline a_{\lambda'''}
  -6\frac{1}{\mathcal N_n}
\sum_{\lambda \in \Lambda_n} |a_{\lambda}|^2
 +3,
 %\end{split}
\end{align}
where the summation with the subscript $ \lambda - \lambda' + \lambda'' - \lambda'''=0$ is over
$(\lambda, -\lambda', \lambda'', -\lambda''')\in S_n(4)$. By the fine structure of $S_n(4)$ described in Lemma \ref{lem:S4 fine struct}, the right-hand side of
 \eqref{eq:int H4 expl} simplifies to
\begin{eqnarray*}
\int_{\mathbb T} H_{4}(T_n(x))\,dx&=&3\frac{1}{\mathcal N_n^2}\Big( \sum_{\lambda, \lambda'\in \Lambda_n} |a_\lambda|^2 |a_{\lambda'}|^2
   -\sum_\lambda |a_{\lambda}|^4   \Big)
  -6\frac{1}{\mathcal N_n}
\sum_{\lambda \in \Lambda_n} |a_{\lambda}|^2 +3 \\
&& =3\frac{1}{\mathcal N_n} \Big( \frac{1}{\sqrt{\mathcal N_n}}
   \sum_{\lambda \in \Lambda_n} (|a_\lambda|^2-1 ) \Big)^2 -
   3\frac{1}{\mathcal N_n^2}
   \sum_{\lambda \in \Lambda_n} |a_\lambda|^4\\
   && =  \frac{6}{\Nc_{n}}\left (\frac{1}{\sqrt{\Nc_{n}/2}}\sum_{\lambda
:\lambda _{2}> 0}(|a_{\lambda }|^{2}-1) +o_{\P}(1) \right
)^{2}-\frac{3}{\Nc_{n}^{2}}\sum_{\lambda }|a_{\lambda }|^{4},
\end{eqnarray*}
where $o_\P(1)=0$ if $n^{1/2}$ is not an integer, otherwise
 $$
o_{\P}(1) =   (\Nc_{n_j}/2)^{-1/2} (|a_{(n^{1/2}, 0)}|^2-1),
$$ thus yielding \paref{aaa} immediately.
The limit \paref{bbb} follows from a standard application of the law of large numbers to the sum,
$$
\frac{3}{\Nc_{n_j}}\sum_{\lambda }|a_{\lambda }|^{4} =\frac{3}{\Nc_{n_j}/2}\sum_{\lambda : \lambda_2> 0}|a_{\lambda }|^{4} + o_{\P}(1),
$$
as well all the variables $a_{\lambda}$ are i.i.d with $$\E\left[ |a_\lambda|^4\right] =2.$$
\end{proof}
\begin{lemma}\label{lem2} For $\ell =1,2$,
$$\displaylines{
\int_{\mathbb T} H_{4}(\partial_\ell \widetilde T_n(x))\,dx
= \frac{24}{\Nc_{n}}\left( \frac{1}{\sqrt{\Nc_{n}/2}}%
\sum_{\lambda ,\lambda _{2}> 0}\left (\frac{\lambda _{\ell}^{2}}{n}\left( |a_{\lambda
}|^{2}-1\right) \right) + o_{\P}(1) \right) ^{2} -\left( \frac{2}{n}\right) ^{2}\frac{3}{\Nc_{n}^{2}}\sum_{\lambda }\lambda
_{\ell}^{4}|a_{\lambda }|^{4}.}$$
Moreover, as $n_j\to \infty$,
\[
\left( \frac{2}{n_j}\right) ^{2}\frac{3}{\Nc_{n_j}}\sum_{\lambda }\lambda
_{\ell}^{4}|a_{\lambda }|^{4} \stackrel{\P}{\longrightarrow} 3(3+\eta).
\]
\end{lemma}
\noindent\begin{proof}
The proof is similar to that of Lemma \ref{lem1}. We have that
\begin{eqnarray}\nonumber
&&\int_{\mathbb T} H_{4}(\partial_\ell \widetilde T_n(x))\,dx
=\int_{\mathbb T}
(\partial_\ell \widetilde T_n(x)^4 -6 \partial_\ell \widetilde T_n(x)^2+3)\,dx\\ \nonumber
&& = \frac{1}{\mathcal N_n^2}\frac{4}{n^2}
\sum_{\lambda, \dots, \lambda''' \in \Lambda_n} \lambda_\ell\lambda'_\ell\lambda''_\ell\lambda'''_\ell a_{\lambda} \overline a_{\lambda'} a_{\lambda''}
\overline a_{\lambda'''}
\int_{\mathbb T}
\exp(2\pi i\langle \lambda -\lambda' + \lambda'' - \lambda''', x\rangle )\,dx+\\ \nonumber
&&  -6\frac{1}{\mathcal N_n}\frac{2}{n}
\sum_{\lambda, \lambda' } \lambda_\ell\lambda'_\ell a_{\lambda}
\overline a_{\lambda'}
\int_{\mathbb T}\exp(2\pi i\langle \lambda - \lambda', x\rangle )\,dx   +3 \\ \nonumber
&&=\frac{1}{\mathcal N_n^2}\frac{4}{n^2}
\sum_{\lambda-\lambda'+\lambda''- \lambda''' =0}\lambda_\ell\lambda'_\ell\lambda''_\ell\lambda'''_\ell a_{\lambda} \overline a_{\lambda'} a_{\lambda''}
\overline a_{\lambda'''}
  -6\frac{1}{\mathcal N_n}\frac{2}{n}
\sum_{\lambda \in \Lambda_n} \lambda_\ell^2 |a_{\lambda}|^2
   +3\\ \nonumber
&&=\frac{3}{\mathcal N_n^2}\frac{4}{n^2}\left (
\sum_{\lambda,\lambda'}\lambda_\ell^2(\lambda'_\ell)^2 |a_{\lambda}|^2 |a_{\lambda'}|^2 - \sum_{\lambda} \lambda_\ell^4 |a_{\lambda}|^4\right )
  -6\frac{1}{\mathcal N_n}\frac{2}{n}
\sum_{\lambda \in \Lambda_n}\lambda_\ell^2 |a_{\lambda}|^2
   +3\\ \label{ciccia}
&& =\frac{24}{\Nc_{n}}\left[ \frac{1}{\sqrt{\Nc_{n}/2}}%
\sum_{\lambda ,\lambda _{2}> 0}\left (\frac{\lambda _{\ell}^{2}}{n}\left( |a_{\lambda
}|^{2}-1\right) \right) + o_{\P}(1) \right] ^{2} -\left( \frac{2}{n}\right) ^{2}\frac{3}{\Nc_{n}^{2}}\sum_{\lambda }\lambda
_{\ell}^{4}|a_{\lambda }|^{4}.
\end{eqnarray}
To conclude the proof, we first observe that the last term in the rhs of \paref{ciccia} may be written as
\begin{eqnarray}
\nonumber
&&\left( \frac{2}{n_j}\right) ^{2}\frac{3}{\Nc_{n_j}}\sum_{\lambda }\lambda
_{\ell}^{4}|a_{\lambda }|^{4} \\ 	\label{avoja}
&&= o_{\P}(1)+ \underbrace{\left( \frac{2}{n_j}\right) ^{2}\frac{3}{\Nc_{n_j}/2}\sum_{\lambda:\lambda_2>0 }\lambda
_{\ell}^{4}(|a_{\lambda }|^{4} -2)}_{=: K_1(n_j)} +\underbrace{\frac{ 24}{n^2_j\Nc_{n_j} }\sum_{\lambda }\lambda
_{\ell}^{4}}_{=:K_2(n_j)}.
\end{eqnarray}
Now for the last term in the rhs of \paref{avoja} we have from Lemma \ref{semplice}
$$
K_2(n_j) = 3 (3+\widehat \mu_n(4)),
$$
so that the conclusion follows from the fact that $\widehat \mu_n(4)\to \eta$, as well as from the fact that, since the random variables $\{|a_\lambda|^4 -2: \lambda\in \Lambda_{n_j}, \, \lambda_2> 0\}$ in $K_1(n_j)$ are i.i.d., square-integrable and centered and $\lambda_\ell^4/n^2\leq 1$, $\E K_1(n_j)^2 = O(\Nc_{n_j}^{-1})\to 0$.

\end{proof}

\medskip

\begin{lemma}\label{lem3} One has that
\begin{eqnarray}\label{e:barber}
&& \int_{\mathbb T} H_2(T_n(x))\Big( H_2(\partial_1 \widetilde
T_n(x)) + H_2(\partial_2 \widetilde
T_n(x))\Big)\,dx\\
&&\hspace{3cm}= \frac{4}{\Nc_{n}}\left( \frac{1}{\sqrt{\Nc_{n}/2}}\sum_{\lambda
,\lambda _{2}> 0}\left( |a_{\lambda }|^{2}-1\right) +o_{\P}(1) \right) ^{2}-\frac{2%
}{\Nc_{n}^{2}}\sum_{\lambda }|a_{\lambda ^{\prime }}|^{4}.\notag
\end{eqnarray}
\end{lemma}
\noindent \begin{proof} For $\ell=1,2$,

%\begin{verbatim}
%Throughout the proof add numbered formulas,
%and clarify what exactly we substitute and where.
%Don't start formulas with ``=..." but always add a lhs.
%\end{verbatim}

\begin{eqnarray}\nonumber
&& \int_{\mathbb T} H_2(T_n(x))H_2(\partial_\ell \widetilde
T_n(x))\,dx=\int_{\mathbb T} (T_n(x)^2-1)(\partial_\ell \widetilde
T_n(x)^2-1)\,dx\\ \nonumber
&&=\int_{\mathbb T} \left (\frac{1}{\mathcal N_n} \sum_{\lambda, \lambda'} a_\lambda \overline{a}_{\lambda'} e_\lambda(x) e_{-\lambda'}(x)-1\right )\left (\frac{2}{n}\frac{1}{\mathcal N_n} \sum_{\lambda'', \lambda'''} \lambda_\ell'' \lambda_\ell''' a_{\lambda''} \overline{a}_{\lambda'''} e_{\lambda''}(x) e_{-\lambda'''}(x)-1 \right )\,dx\\
\label{core}
&& =\frac{2}{n}\frac{1}{\mathcal N_n^2} \sum_{\lambda-\lambda'+\lambda''-\lambda'''=0} \lambda_\ell'' \lambda_\ell''' a_\lambda \overline{a}_{\lambda'} a_{\lambda''} \overline{a}_{\lambda'''}-
\frac{1}{\mathcal N_n} \sum_{\lambda} |a_\lambda|^2
-
\frac{2}{n}\frac{1}{\mathcal N_n} \sum_{\lambda}\lambda_\ell^2 |a_\lambda|^2 +1.
\end{eqnarray}
An application of the inclusion-exclusion principle yields that the first summand in the rhs of \paref{core} equals
\begin{eqnarray}\label{buondi}
&&\frac{2}{n}\frac{1}{\mathcal N_n^2} \sum_{\lambda-\lambda'+\lambda''-\lambda'''=0} \lambda_\ell'' \lambda_\ell''' a_\lambda \overline{a}_{\lambda'} a_{\lambda''} \overline{a}_{\lambda'''} \\  \nonumber
&&=
\frac{2}{n}\frac{1}{\mathcal N_n^2} \left ( \sum_{\lambda,\lambda'} \lambda_j^2 |a_\lambda|^2 |a_{\lambda'}|^2 +2\sum_{\lambda,\lambda'} \lambda_j \lambda'_j |a_\lambda|^2 |a_{\lambda'}|^2 - 2\sum_\lambda \lambda_j^2|a_\lambda|^4 + \sum_\lambda \lambda_j^2|a_\lambda|^4 \right ).
\end{eqnarray}
Using the relation $a_{-\lambda} = \overline a_\lambda$, we also infer that
\begin{equation}\label{zeroo}
\sum_{\lambda,\lambda'} \lambda_j \lambda'_j |a_\lambda|^2 |a_{\lambda'}|^2 = \left (\sum_\lambda \lambda_j |a_\lambda|^2 \right)^2 =0.
\end{equation}
Substituting \paref{zeroo} into  \paref{buondi} and then \paref{buondi} into \paref{core} we rewrite \paref{core} as
\begin{eqnarray}\nonumber
&& \int_{\mathbb T} H_2(T_n(x))H_2(\partial_\ell \widetilde
T_n(x))\,dx\\
\label{core2}
&& =\frac{2}{n}\frac{1}{\mathcal N_n^2} \left ( \sum_{\lambda,\lambda'} \lambda_j^2 |a_\lambda|^2 |a_{\lambda'}|^2  - \sum_\lambda \lambda_j^2|a_\lambda|^4  \right ) -
\frac{1}{\mathcal N_n} \sum_{\lambda} |a_\lambda|^2
\\ \nonumber
&&-
\frac{2}{n}\frac{1}{\mathcal N_n} \sum_{\lambda}\lambda_\ell^2 |a_\lambda|^2 +1.
\end{eqnarray}

Summing the terms corresponding to $\partial _{1}$ and $\partial _{2}$ up, i.e.  \paref{core2} for $\ell=1$ and $\ell=2$, we deduce that the lhs of \eqref{e:barber} equals
\begin{eqnarray*}
&&\int_{\mathbb T} H_2(T_n(x))\Big( H_2(\partial_1 \widetilde
T_n(x)) + H_2(\partial_2 \widetilde
T_n(x))\Big)\,dx\\
&&=\frac{2}{n}\frac{1}{\Nc_{n}^{2}}\left( \sum_{\lambda ,\lambda ^{\prime
}}\left( \lambda _{1}^{2}+\lambda _{2}^{2}\right) |a_{\lambda
}|^{2}|a_{\lambda ^{\prime }}|^{2}-\sum_{\lambda }\left( \lambda
_{1}^{2}+\lambda _{2}^{2}\right) |a_{\lambda ^{\prime }}|^{4}\right) \\
&&-\frac{2}{\Nc_{n}}\sum_{\lambda }|a_{\lambda }|^{2}-\frac{2}{n}\frac{1}{\Nc_{n}%
}\sum_{\lambda }\left( \lambda _{1}^{2}+\lambda _{2}^{2}\right) |a_{\lambda
}|^{2}+2\\
&&=\frac{2}{n}\frac{1}{\Nc_{n}^{2}}\left( \sum_{\lambda ,\lambda ^{\prime
}}n|a_{\lambda }|^{2}|a_{\lambda ^{\prime }}|^{2}-\sum_{\lambda
}n|a_{\lambda ^{\prime }}|^{4}\right) -\frac{2}{\Nc_{n}}\sum_{\lambda }|a_{\lambda }|^{2}-\frac{2}{n}\frac{1}{\Nc_{n}%
}\sum_{\lambda }n|a_{\lambda }|^{2}+2 \\
&&=2\frac{1}{\Nc_{n}^{2}}\left( \sum_{\lambda ,\lambda ^{\prime }}|a_{\lambda
}|^{2}|a_{\lambda ^{\prime }}|^{2}-\sum_{\lambda }|a_{\lambda ^{\prime
}}|^{4}\right) -\frac{2}{\Nc_{n}}\sum_{\lambda }|a_{\lambda }|^{2}-\frac{2}{\Nc_{n}}%
\sum_{\lambda }|a_{\lambda }|^{2}+2\\
&&=\frac{2}{\Nc_{n}}\left( \frac{\sqrt{2}}{\sqrt{\Nc_{n}/2}}\sum_{\lambda
,\lambda _{2}> 0}\left( |a_{\lambda }|^{2}-1\right) +o_{\P}(1)  \right ) ^{2}-\frac{2%
}{\Nc_{n}^{2}}\sum_{\lambda }|a_{\lambda ^{\prime }}|^{4},
\end{eqnarray*}
which equals to the rhs of \paref{e:barber}.
\end{proof}

\medskip

Our last lemma allows one to deal with the most challenging term appearing in \eqref{e:nus}.

\begin{lemma}\label{lem4}
We have that
\begin{eqnarray*}
 &&\int H_{2}(\partial _{1}\widetilde{T}_{n})H_{2}(\partial _{2}\widetilde{T}%
_{n})\,dx  \\
&&=-4\left[ \frac{1}{\sqrt{\Nc_{n}/2}}\frac{1}{n}\sum_{\lambda ,\lambda _{2}>
0}\lambda _{2}^{2}(|a_{\lambda }|^{2}-1) \right] ^{2}
-4\left[ \frac{1}{\sqrt{\Nc_{n}/2}}\frac{1}{n}\sum_{\lambda ,\lambda _{2}>
0}\lambda _{1}^{2}(|a_{\lambda }|^{2}-1)+o_{\P}(1)\right] ^{2}\\
&&\quad +4\left[  \frac{1}{\sqrt{\Nc_{n}/2}}\sum_{\lambda ,\lambda _{2}>
0}(|a_{\lambda }|^{2}-1) +o_{\P}(1)
\right] ^{2} \\
&& \quad+16\left[
\frac{1}{\sqrt{\Nc_{n}/2}}\frac{1}{n}\sum_{\lambda ,\lambda _{2}>
0}\lambda _{1}\lambda _{2}\left( |a_{\lambda }|^{2}-1\right)+o_{\P}(1)
\right] ^{2}
-\frac{12}{n^{2}}\frac{1}{\Nc_{n}^{2}}\sum_{\lambda
}\lambda _{1}^{2}\lambda _{2}^{2}|a_{\lambda }|^{4},
\end{eqnarray*}
and the following convergence takes place as $n_j\to\infty$:
\begin{equation}\label{torna a surriento}
\frac{12}{n_j^{2}}\frac{1}{\Nc_{n_j}^{2}}\sum_{\lambda
}\lambda _{1}^{2}\lambda _{2}^{2}|a_{\lambda }|^{4}\stackrel{\P}{\longrightarrow} 3(1-\eta).
\end{equation}

\end{lemma}

\noindent\begin{proof} One has that
\begin{equation}
\label{hope1}
\begin{split}
\int H_{2}(\partial _{1}\widetilde{T}_{n})H_{2}(\partial _{2}\widetilde{T}%
_{n})dx
&=\frac{4}{n^{2}}\frac{1}{\Nc_{n}^{2}}\sum_{\lambda -\lambda ^{\prime
}+\lambda ^{\prime \prime }-\lambda ^{\prime \prime \prime }=0}\lambda
_{1}\lambda _{1}^{\prime }\lambda _{2}^{\prime \prime }\lambda _{2}^{\prime
\prime \prime }a_{\lambda }\overline{a}_{\lambda ^{\prime }}a_{\lambda
^{\prime \prime }}\overline{a}_{\lambda ^{\prime \prime \prime }} \\
&-\frac{2}{n}\frac{1}{\Nc_{n}}\sum_{\lambda }\lambda _{1}^{2}|a_{\lambda
}|^{2}-\frac{2}{n}\frac{1}{\Nc_{n}}\sum_{\lambda }\lambda _{2}^{2}|a_{\lambda
}|^{2}+1.
\end{split}
\end{equation}%
First of all, we note that for the first two terms in \paref{hope1}%
\[
\mathbb{E}\left[ \frac{2}{n}\frac{1}{\Nc_{n}}\sum_{\lambda }(\lambda
_{1}^{2}+\lambda _{2}^{2})|a_{\lambda }|^{2}\right] =\mathbb{E}\left[ \frac{2%
}{\Nc_{n}}\sum_{\lambda }|a_{\lambda }|^{2}\right] =2.
\]%
Let us now focus on \eqref{hope1}. Using the structure of $S_4(n)$ in Lemma \ref{lem:S4 fine struct}, we obtain
\begin{eqnarray}\label{viva}
&&\frac{4}{n^{2}}\frac{1}{\Nc_{n}^{2}}\sum_{\lambda -\lambda ^{\prime }+\lambda
^{\prime \prime }-\lambda ^{\prime \prime \prime }=0}\lambda _{1}\lambda
_{1}^{\prime }\lambda _{2}^{\prime \prime }\lambda _{2}^{\prime \prime
\prime }a_{\lambda }\overline{a}_{\lambda ^{\prime }}a_{\lambda ^{\prime
\prime }}\overline{a}_{\lambda ^{\prime \prime \prime}}
\\ \nonumber
&&=\frac{4}{n^{2}}\frac{1}{\Nc_{n}^{2}}\left[ \sum_{\lambda ,\lambda
^{\prime }}\lambda _{1}^{2}(\lambda _{2}^{\prime })^{2}|a_{\lambda
}|^{2}|a_{\lambda ^{\prime }}|^{2}+2\sum_{\lambda ,\lambda
^{\prime }}\lambda _{1}\lambda _{2}\lambda _{1}^{\prime }\lambda
_{2}^{\prime }|a_{\lambda }|^{2}|a_{\lambda ^{\prime
}}|^{2}-3\sum_{\lambda }\lambda _{1}^{2}\lambda
_{2}^{2}|a_{\lambda }|^{4}\right].
\end{eqnarray}
Let us now denote
\begin{eqnarray}\nonumber
A &:=& \frac{4}{n^{2}}\frac{1}{\Nc_{n}^{2}}\sum_{\lambda ,\lambda ^{\prime }}\lambda
_{1}^{2}(\lambda _{2}^{\prime })^{2}|a_{\lambda }|^{2}|a_{\lambda ^{\prime
}}|^{2}, \\ \nonumber
B &:=& \frac{4}{n^{2}}\frac{1}{\Nc_{n}^{2}}2\sum_{\lambda ,\lambda ^{\prime }}\lambda
_{1}\lambda _{2}\lambda _{1}^{\prime }\lambda _{2}^{\prime }|a_{\lambda
}|^{2}|a_{\lambda ^{\prime }}|^{2}, \\ \label{cC}
C &:=& -3\frac{4}{n^{2}}\frac{1}{\Nc_{n}^{2}}\sum_{\lambda }\lambda _{1}^{2}\lambda
_{2}^{2}|a_{\lambda }|^{4}, \\ \label{Dd}
D &:=& \frac{4}{n^{2}}\frac{1}{\Nc_{n}^{2}}\left\{ -N_{n}\frac{n}{2}\sum_{\lambda
}|a_{\lambda }|^{2}+\frac{\Nc_{n}^{2}n^{2}}{4}\right\},
\end{eqnarray}%
so that \eqref{hope1} with \eqref{viva} read
\begin{equation}
\label{sommatot}
\int H_{2}(\partial _{1}\widetilde{T}_{n})H_{2}(\partial _{2}\widetilde{T}%
_{n})\,dx = A + B + C + D.
\end{equation}

We have that %$A$ equals%
\begin{eqnarray*}
A  && =\frac{4}{n^{2}}\frac{1}{\Nc_{n}^{2}}\sum_{\lambda ,\lambda ^{\prime
}}\lambda _{1}^{2}(\lambda _{2}^{\prime })^{2}|a_{\lambda }|^{2}|a_{\lambda
^{\prime }}|^{2} \\
&&=\frac{4}{n^{2}}\frac{1}{\Nc_{n}^{2}}\frac{1}{2}\left\{ \sum_{\lambda
,\lambda ^{\prime }}\lambda _{1}^{2}(\lambda _{2}^{\prime })^{2}|a_{\lambda
}|^{2}|a_{\lambda ^{\prime }}|^{2}+\sum_{\lambda ,\lambda ^{\prime }}\lambda
_{1}^{2}(\lambda _{2}^{\prime })^{2}|a_{\lambda }|^{2}|a_{\lambda ^{\prime
}}|^{2}\right\}
\\
&&=\frac{4}{n^{2}}\frac{1}{\Nc_{n}^{2}}\frac{1}{2}\left\{ \sum_{\lambda ,\lambda
^{\prime }}(n-\lambda _{2}^{2})(\lambda _{2}^{\prime })^{2}|a_{\lambda
}|^{2}|a_{\lambda ^{\prime }}|^{2}+\sum_{\lambda ,\lambda ^{\prime }}\lambda
_{1}^{2}(n-(\lambda _{1}^{\prime })^{2})|a_{\lambda }|^{2}|a_{\lambda
^{\prime }}|^{2}\right\},
\end{eqnarray*}
which we may rewrite as
\begin{eqnarray}\nonumber
A &&=\frac{4}{n^{2}}\frac{1}{\Nc_{n}^{2}}\frac{1}{2}\left\{ -\sum_{\lambda
,\lambda ^{\prime }}\lambda _{2}^{2}(\lambda _{2}^{\prime })^{2}|a_{\lambda
}|^{2}|a_{\lambda ^{\prime }}|^{2}-\sum_{\lambda ,\lambda ^{\prime }}\lambda
_{1}^{2}(\lambda _{1}^{\prime })^{2}|a_{\lambda }|^{2}|a_{\lambda ^{\prime
}}|^{2}\right\}  \\ \nonumber
&&+\frac{4}{n^{2}}\frac{1}{\Nc_{n}^{2}}\frac{1}{2}\left\{ n\sum_{\lambda
,\lambda ^{\prime }}(\lambda _{2}^{\prime })^{2}|a_{\lambda
}|^{2}|a_{\lambda ^{\prime }}|^{2}+n\sum_{\lambda ,\lambda ^{\prime
}}\lambda _{1}^{2}|a_{\lambda }|^{2}|a_{\lambda ^{\prime }}|^{2}\right\}
\\ \nonumber
&&=\frac{4}{n^{2}}\frac{1}{\Nc_{n}^{2}}\frac{1}{2}\left\{ -\sum_{\lambda
,\lambda ^{\prime }}\lambda _{2}^{2}(\lambda _{2}^{\prime })^{2}|a_{\lambda
}|^{2}|a_{\lambda ^{\prime }}|^{2}-\sum_{\lambda ,\lambda ^{\prime }}\lambda
_{1}^{2}(\lambda _{1}^{\prime })^{2}|a_{\lambda }|^{2}|a_{\lambda ^{\prime
}}|^{2}\right\}  \\ \label{ollala}
&&+\frac{4}{n^{2}}\frac{1}{\Nc_{n}^{2}}\frac{1}{2}\left\{ n\sum_{\lambda
,\lambda ^{\prime }}(\lambda _{2}^{\prime })^{2}|a_{\lambda
}|^{2}|a_{\lambda ^{\prime }}|^{2}+n\sum_{\lambda ,\lambda ^{\prime
}}\lambda _{1}^{2}|a_{\lambda }|^{2}|a_{\lambda ^{\prime }}|^{2}\right\}.
\end{eqnarray}%
From \paref{Dd} and \paref{ollala}, we get
\begin{eqnarray}\nonumber
A+D &=&-\frac{4}{n^{2}}\frac{1}{\Nc_{n}^{2}}\frac{1}{2}\left[ \sum_{\lambda
}\lambda _{2}^{2}(|a_{\lambda }|^{2}-1)\right] ^{2} -\frac{4}{n^{2}}\frac{1}{\Nc_{n}^{2}}\frac{1}{2}\left[ \sum_{\lambda
}\lambda _{1}^{2}(|a_{\lambda }|^{2}-1)\right] ^{2} \\ \nonumber
&&+\frac{4}{\Nc_{n}^{2}}\frac{1}{2}\left[ \sum_{\lambda }(|a_{\lambda }|^{2}-1)%
\right] ^{2}\\  \label{a+d}
&=&-4\left[ \frac{1}{\sqrt{\Nc_{n}/2}}\frac{1}{n}\sum_{\lambda ,\lambda _{2}>
0}\lambda _{2}^{2}(|a_{\lambda }|^{2}-1) +o_\P(1)\right] ^{2}\\ \nonumber
&&
-4\left[ \frac{1}{\sqrt{\Nc_{n}/2}}\frac{1}{n}\sum_{\lambda ,\lambda _{2}>
0}\lambda _{1}^{2}(|a_{\lambda }|^{2}-1)+o_{\P}(1)\right] ^{2}\\
\nonumber
&&\quad +4\left[  \frac{1}{\sqrt{\Nc_{n}/2}}\sum_{\lambda ,\lambda _{2}>
0}(|a_{\lambda }|^{2}-1) +o_{\P}(1)
\right] ^{2}.
\end{eqnarray}%
On the other hand,
\begin{eqnarray}\nonumber
B &=&\frac{4}{n^{2}}\frac{1}{\Nc_{n}^{2}}2\sum_{\lambda ,\lambda ^{\prime
}}\lambda _{1}\lambda _{2}\lambda _{1}^{\prime }\lambda _{2}^{\prime
}|a_{\lambda }|^{2}|a_{\lambda ^{\prime }}|^{2}
=\frac{4}{n^{2}}\frac{1}{\Nc_{n}^{2}}2\left[ \sum_{\lambda }\lambda
_{1}\lambda _{2}\left( |a_{\lambda }|^{2}-1\right) \right] ^{2}\\ \label{bB}
&=&   16\left[
\frac{1}{\sqrt{\Nc_{n}/2}}\frac{1}{n}\sum_{\lambda ,\lambda _{2}>
0}\lambda _{1}\lambda _{2}\left( |a_{\lambda }|^{2}-1\right)+o_{\P}(1)
\right] ^{2}.
\end{eqnarray}%
The first statement of Lemma \ref{lem4} then follows upon substituting \paref{a+d}, \paref{bB} and \paref{cC} into \paref{sommatot}.

Now to prove \paref{torna a surriento} it suffices to write
\begin{eqnarray}\nonumber
&&\frac{12}{{n_j}^{2}}\frac{1}{\Nc_{n_j}^{2}}\sum_{\lambda }\lambda _{1}^{2}\lambda
_{2}^{2}|a_{\lambda }|^{4}\\
&&=\frac{12}{{n_j}^{2}}\frac{1}{\Nc_{n_j}^{2}}\sum_{\lambda }\lambda _{1}^{2}\lambda
_{2}^{2}(|a_{\lambda }|^{4}-2)  +\frac{12}{{n_j}^{2}}\frac{2}{\Nc_{n_j}^{2}}\sum_{\lambda }\lambda _{1}^{2}\lambda
_{2}^{2}\\ \nonumber
&&= o_\P(1) + \frac{12}{{n_j}^{2}}\frac{2}{\Nc_{n_j}^{2}}\sum_{\lambda: \lambda_2>0 }\lambda _{1}^{2}\lambda
_{2}^{2}(|a_{\lambda }|^{4}-2)  +\frac{24}{{n_j}^{2}}\frac{1}{\Nc_{n_j}^{2}}\sum_{\lambda }\lambda _{1}^{2}\lambda
_{2}^{2},
\end{eqnarray}
and then Lemma \ref{semplice} and an argument similar to the one that concluded the proof of Lemma \ref{lem2} allow to prove the result.

\end{proof}

\subsection{Proof of Lemma \ref{p:formula4}: technical computations}

Substituting the results of Lemmas \ref{lem1}-\ref{lem4} into \eqref{e:nus} we obtain, as $n_j\to +\infty$,
\begin{equation}
\label{e:nus2}
\begin{split}
\mathcal{L}_{n_j}[4]&=\sqrt{\frac{E_{n_j}}{2\Nc^2_{n_j}}}\Bigg(\frac{\alpha_{0,0}\beta_4}{4!}6
\bigg(\frac{1}{\sqrt{\Nc_{n_j}/2}}\sum_{\lambda :\lambda
_{2}> 0}(|a_{\lambda }|^{2}-1)  \bigg)^{2} \\
&+\frac{\alpha_{0,4}\beta_0}{4!}\cdot 24\bigg( \frac{1}{\sqrt{\Nc_{n_j}/2}}%
\sum_{\lambda ,\lambda _{2}> 0}\left (\frac{\lambda _{1}^{2}}{n_j}\left( |a_{\lambda
}|^{2}-1\right) \right)  \bigg) ^{2}\\
&+\frac{\alpha_{4,0}\beta_0}{4!}\cdot 24\bigg( \frac{1}{\sqrt{\Nc_{n_j}/2}}%
\sum_{\lambda ,\lambda _{2}> 0}\bigg(\frac{\lambda _{2}^{2}}{n_j}\left( |a_{\lambda
}|^{2}-1\right) \bigg) \bigg) ^{2}\\
+\frac{\alpha_{0,2}\beta_2}{2!2!}\cdot 4\bigg ( &\frac{1}{\sqrt{\Nc_{n_j}/2}}\sum_{\lambda
,\lambda _{2}> 0}\left( |a_{\lambda }|^{2}-1\right) \bigg )^{2} +\frac{\alpha_{2,2}\beta_0}{2!2!}\Bigg(-4\bigg( \frac{1}{\sqrt{\Nc_{n_j}/2}}\frac{1}{n_j}\sum_{\lambda ,\lambda _{2}\geq
0}\lambda _{2}^{2}(|a_{\lambda }|^{2}-1) \bigg) ^{2}\\
&-4\bigg( \frac{1}{\sqrt{\Nc_{n_j}/2}}\frac{1}{n_j}\sum_{\lambda ,\lambda _{2}>
0}\lambda _{1}^{2}(|a_{\lambda }|^{2}-1)\bigg) ^{2} +4\bigg(  \frac{1}{\sqrt{\Nc_{n_j}/2}}\sum_{\lambda ,\lambda _{2}>
0}(|a_{\lambda }|^{2}-1)
\bigg) ^{2} \\
& \quad+16\bigg(
\frac{1}{\sqrt{\Nc_{n_j}/2}}\frac{1}{n_j}\sum_{\lambda ,\lambda _{2}>
0}\lambda _{1}\lambda _{2}\bigg( |a_{\lambda }|^{2}-1\bigg)
\bigg) ^{2}\Bigg ) + \widetilde R(n_j)\Bigg),
\end{split}
\end{equation}
where $\widetilde R(n_j)$ is a sequence of random variables converging in probability to some constant $\in \R$. Computing the coefficients $\alpha_{\cdot, \cdot}$ (see \eqref{e:alpha}) and $\beta_{\cdot}$ (see \eqref{e:beta}), from \paref{e:nus2} we obtain that, as $n_j\to +\infty$,
\begin{equation*}
\mathcal{L}_{n_j}[4]=\sqrt{\frac{E_{n_j}}{ 512\, \mathcal N^2_{n_j}}}\Big( W_1(n_j) ^2-2W_2(n_j)^2-2W_3(n_j)^2 - 4W_4(n_j)^2+ R(n_j) \Big),
\end{equation*}
where $W_k(n_j)$, $k=1,2,3,4$ are as in \paref{Wdef} and $R(n_j)$ is a sequence of random variables converging in probability to $1$. The proof is now complete.

\qed

%\end{proof}

\bigskip


\begin{thebibliography}{99}

\bibitem[AS]{AS} Abramowitz, M.; Stegun, I. A. (1964) \emph{Handbook of mathematical functions with formulas, graphs, and
              mathematical tables}, National Bureau of Standards Applied Mathematics Series. Washington, D.C.


%\bibitem[AT]{adlertaylor} Adler, R. J.; Taylor, J. E. (2007) \emph{Random Fields and Geometry}, Springer Monographs in %Mathematics. Springer, New York.

\bibitem[Be]{Berry 2002} Berry, M.V. (2002) Statistics of nodal lines and points in chaotic quantum billiards:
perimeter corrections, fluctuations, curvature, \emph{Journal of Physics A}, 35, 3025-3038

\bibitem[CMW]{cmw2015} Cammarota, V.; Marinucci, D.; Wigman, I. (2015) Fluctuations of the Euler-Poncar\'e characteristic for random spherical
harmonics, \emph{Proceedings of the American Mathematical Society}, in press, arXiv:1504.01868

\bibitem[Ci]{Ci} Cilleruelo, J. (1993) The distribution of the lattice points on circles, \emph{Journal of Number Theory}
43, no. 2, 198--202.

\bibitem[DF]{DF} Donnelly, H; Fefferman, C. (1988) Nodal sets of eigenfunctions on Riemannian manifolds,
\emph{Inventiones Mathematicae}, 93, 161--183.

\bibitem[Du]{D} R. M. Dudley (2002) \emph{Real analysis and probability}. Cambridge University Press.

\bibitem[EH]{EH} Erdos, P.; Hall, R.R. (1999) On the angular distribution of Gaussian integers with fixed
norm, \emph{Discrete Mathematics}, 200(1--3):87--94, Paul Erd\"{o}s memorial collection.

\bibitem[KL]{KL} Kratz, M. F.; Le{\'o}n, J. R. (2001) Central limit theorems for level functionals of stationary {G}aussian processes and fields, \emph{Journal of Theoretical Probability,} 14, no. 3, 639--672.

\bibitem[KKW]{KKW} Krishnapur, M., Kurlberg, P., and Wigman, I. (2013) Nodal length fluctuations for arithmetic random waves, \emph{Annals of Mathematics}
(2) 177, no. 2, 699--737.

\bibitem[KW]{KW} Kurlberg, P.; Wigman, I. (2015) On probability measures arising from lattice points on circles, \emph{Mathematische Annalen}, in press, arXiv:1501.01995

\bibitem[La]{La} Landau, E. (1908) Uber die Einteilung der positiven Zahlen nach vier Klassen nach der
Mindestzahl der zu ihrer addition Zusammensetzung erforderlichen Quadrate, \emph{Archiv der Mathematik und Physics} III.

\bibitem[MR]{mrossi} Marinucci, D.; Rossi, M. (2015) Stein-Malliavin approximations for nonlinear functionals of random eigenfunctions on $S^d$,
\emph{Journal of Functional Analysis}, 268, n.8, 2379-2420

\bibitem[MW]{MaWi1} Marinucci, D.; Wigman, I. (2011) On the excursion sets of spherical Gaussian eigenfunctions, \emph{Journal of Mathematical Physics},
52, 093301, arXiv 1009.4367

\bibitem[NP]{N-P} Nourdin, I.; Peccati, G. (2012) \emph{Normal approximations with Malliavin calculus. From Stein's method to universality}. Cambridge University Press.

\bibitem[NR]{N-R} Nourdin, I.; Rosinski, J. (2014) Asymptotic independence of
multiple Wiener-It\^o integrals and the resulting limit laws, \emph{Annals of Probability}, 42, no. 2, 497-526.

\bibitem[PT]{P-T} Peccati, G.; Taqqu, M.S. (2010) \emph{Wiener Chaos: Moments, Cumulants and Diagrams}. Springer-Verlag.

\bibitem[Ro]{rossiphd} Rossi, M. (2015) \emph{The Geometry of Spherical Random Fields}, PhD thesis, University of Rome Tor Vergata. ArXiv 1603.07575.

\bibitem[RW]{RudWig} Rudnick, Z.; Wigman, I. (2008) On the volume of nodal sets for eigenfunctions of the Laplacian on the torus, \emph{Annales de l'Insitute Henri Poincar\'e},  9,  no. 1, 109--130.

%\bibitem[RW2]{RW} Rudnick, Z.; Wigman, I. (2014) Nodal intersections for random eigenfunctions on the torus, \emph{American Journal of Mathematics}, in press, arXiv:1402.3621

\bibitem[vdV]{VV} van der Vaart, A. W. (1998) \emph{Asymptotic statistics}. Cambridge University Press.

\bibitem[W1]{wig} Wigman, I. (2010)  Fluctuations of the nodal length of random spherical harmonics,
\emph{Communications in Mathematical Physics}, 298(3):787--831

\bibitem[W2]{wigsurvey} Wigman, I. (2012)  On the nodal lines of random and deterministic Laplace eigenfunctions.
 \emph{Proceedings of the International Conference on Spectral Geometry, Dartmouth College}, 84:285-298

\bibitem[Ya]{Yau} Yau, S.T. (1982) Survey on partial differential equations in differential geometry, \emph{Seminar on Differential Geometry}, pp. 3--71, Annals of Mathematics Studies, 102, Princeton University Press.

%\bibitem[Ya2]{Yau2} Yau, S.T. (1993) Open problems in geometry. Differential geometry: partial differential equations on manifolds, 1--28, {\it Proc. Sympos. Pure Math.}, 54, Part 1, American Mathematical Society, Providence, RI.

\end{thebibliography}
\end{document}